\documentclass[a4paper, 11pt]{article}
\usepackage[left=1.8cm,top=1.9cm, bottom=1in, right=2cm]{geometry}

\usepackage{graphicx,subfigure, mathrsfs}

\usepackage{soul,color,tikz}

 \sethlcolor{yellow}
\usepackage{amsmath}
\usepackage{amssymb}               
\usepackage{amsfonts}               
\usepackage{amsopn,amsthm}
\usepackage{bbm}
\usepackage{tikz}

\usetikzlibrary{fit}					
\usetikzlibrary{backgrounds}	
\newtheorem{theorem}{Theorem}
\newtheorem{lemma}{Lemma}

\newtheorem{corollary}{Corollary}
\newtheorem{definition}{Definition}

\newtheorem{remark}{Remark}

\def\hY{\hat{Y}}

\def\tR{\tilde{R}}

\def\n{\nonumber\\}
\def\be{\begin{equation}}
\def\ee{\end{equation}}
\def\bes{\begin{equation*}}
\def\ees{\end{equation*}}
\def\ma{{\mathcal A}}
\def\mb{{\mathcal B}}

\def\mf{{\mathcal F}}

\def\mm{{\mathcal M}}

\def\mr{{\mathcal R}}
\def\ms{{\mathcal S}}

\def\mw{{\mathcal W}}
\def\mx{{\mathcal X}}
\def\my{{\mathcal Y}}

\def\F{\mathbf{F}}

\def\e{\mathbb{E}}

\def\tv#1{\left\|#1\right\|_1}
\def\apx#1{\stackrel{#1}{\approx}}

\def\F{F_{[1:r]}}

\def\ta{\mathsf{T}_1}
\def\tb{\mathsf{T}_2}
\def\t#1{\mathsf{T}_{(#1)_2}}

\allowdisplaybreaks
\providecommand{\keywords}[1]{\textbf{{Index terms---}} #1}

\begin{document}

\date{}
\author{Mohammad~Hossein~Yassaee, Amin~Gohari, Mohammad~Reza~Aref
\thanks{\scriptsize \noindent 
The authors are with the Information Systems and Security Lab (ISSL), Department of Electrical Engineering, Sharif University of Technology, Tehran, Iran (e-mails: yassaee@ee.sharif.edu; \{aminzadeh,aref\}@sharif.edu). This work was supported by Iran-NSF under grant No. 92-32575. This paper was presented in part at ISIT 2012.}
}

\title{Channel Simulation via Interactive Communications}

\maketitle
\begin{abstract}
In this paper, we study the problem of channel simulation via interactive communication, known as the coordination capacity, in a two-terminal network. We assume that two terminals observe i.i.d.\ copies of two random variables and would like to generate i.i.d.\ copies of two other random variables jointly distributed with the observed random variables. The terminals are provided with two-way communication links, and shared common randomness, all at limited rates. Two special cases of this problem are the interactive function computation studied by Ma and Ishwar, and the tradeoff curve between one-way communication and shared randomness studied by Cuff. The latter work had inspired Gohari and Anantharam  to study the general problem of channel simulation via interactive communication stated above. However only inner and outer bounds for the special case of no shared randomness were obtained in their work. In this paper we settle this problem by providing an exact computable characterization of the multi-round problem. To show this we employ the technique of ``output statistics of random binning" that has been recently developed by the authors.

\end{abstract}
\section{Introduction}
The minimum amount of interaction needed to create dependent random variables is an operational way to quantify the correlation among random variables. Wyner considered
the problem of remote reconstruction of two dependent random variables by two terminals which are provided with shared randomness at a limited rate \cite{Wyner}. He used this approach to measure the intrinsic common randomness between two random variables. An alternative characterization of Wyner's common information as an extreme point of a channel simulation problem was provided in \cite{cuff,cuff-trans}. In this setup, a terminal  observing i.i.d.\ copies of $X$, sends a message at rate $R_1$ to a remote random number generator (decoder) that produces i.i.d.\ copies of another random variable $Y$ which is jointly distributed with $X$. The total variation distance between the achieved joint distribution and the i.i.d.\ distribution induced by passing $X$ through a discrete memoryless channel (DMC) channel $p(y|x)$ should be negligible. In other words, the generated distribution and the i.i.d.\ distribution should  statistically be indistinguishable. Shared common randomness exists between the two parties at a limited rate $R_0$. Cuff found the tradeoff between $R_0$ and $R_1$ showing that when $R_0=0$ the minimum admissible rate for $R_1$ is the Wyner's common information; and when $R_0=\infty$, the minimum admissible rate for $R_1$ is the mutual information between $X$ and $Y$ (this special case was already shown in \cite{Bennett}). 

This setup was  generalized in \cite{aminzade} by assuming that two terminals have access to i.i.d. copies of $X_1$ and $X_2$ respectively and would like to generate i.i.d.\ copies of $Y_1$ and $Y_2$. Instead of a one-way communication, now the terminals are provided with a two-way communication at rates $R_{12}$ and $R_{21}$ (see Fig. 1). They can use up these two resources in $r$ rounds of interactive communications as they wish (i.e. we only impose the constraint that $\sum_{i~odd} H(C_i)$ is less than or equal to $nR_{12}$ where $H(C_i)$ is the entropy of the message sent from terminal 1 to terminal 2 at round $i$; a similar statement holds for $R_{21}$). Inner and outer bounds on $R_{12}$ and $R_{21}$ were derived in the special case of no shared common randomness \cite{aminzade}. In this paper we completely solve this problem under both the \emph{strong} and \emph{empirical} coordination models. Strong coordination demands a total variation converging to zero. On the other hand, empirical coordination only demands closeness of the \emph{empirical distribution} of the generated random variables and the i.i.d.\ ones \cite{coordination} (See Section \ref{sec:II} for a detailed description of these two models).
 %
\begin{figure}[t]
\centering
\tikzstyle{enc}=[rectangle,
                                    thick,
                                    minimum height=5cm,
                                    minimum width=2cm,
                                    draw=blue!80,
                                    fill=blue!20,scale=.5]

\tikzstyle{background}=[rectangle,
                                                fill=yellow!20,
                                                inner sep=0.2cm,
                                                rounded corners=5mm,
                                               scale=.5 ]
\begin{tikzpicture}[>=stealth,scale=.5]
\node (e1) [enc] {\Large{$\mbox{Enc}_1/\mbox{Dec}_1$}};
\path (e1.180)+(8,0) node[enc] (e2) {\Large{$\mbox{Enc}_2/\mbox{Dec}_2$}};
\path (e2.north west)+(0,-.2) node[coordinate](c1){};
\path (e2.north west)+(0,-1.2) node[coordinate](c2){};
\path (e2.north west)+(0,-2.2) node[coordinate](c3){};
\path (e2.south west)+(0,.2) node[coordinate](cr){};
\path [draw,->] (e1.north east)+(0,-.2)--node[above,scale=.7](d0){\textcolor{black}{$C_1$}} (c1);
\path [draw,<-] (e1.north east)+(0,-1.2)--node[above,scale=.7]{\textcolor{black}{$C_2$}} (c2);
\path [draw,->] (e1.north east)+(0,-2.2)--node[above,scale=.7](d1){\textcolor{black}{$C_3$}} (c3);
\path [draw,<-] (e1.south east)+(0,.2)--node[above,scale=.7](dr){\textcolor{black}{$C_r$}}(cr);
\path [draw,densely dashed] (d1)+(0,-.7) -- (dr);
\path (d0)+(0,1.25) node[scale=.8](sim){Channel Simulator: $q(y_1y_2|x_1x_2)$};
\path (d0)+(0,.6) node[scale=.8]{with $r$ communication rounds};
\path (dr)+(0,-2) node[](o){$nR_0$ bits};
\draw[->] (o)--(e1.south east);
\draw[->] (o)--(e2.south west);
\path (e1.north west)+(-1,-1) node[coordinate](s1){};
\path (s1)+(-2,0) node[](x1){$X_1^n$};
\path (e1.south west)+(-1,1) node[coordinate](r1){};
\path (r1)+(-2,0) node[](y1){$Y_1^n$};
\path (e2.north east)+(1,-1) node[coordinate](s2){};
\path (s2)+(2,0) node[](x2){$X_2^n$};
\path (e2.south east)+(1,1) node[coordinate](r2){};
\path (r2)+(2,0) node[](y2){$Y_2^n$};
 \begin{pgfonlayer}{background}
        \node [background,
                    fit=(s1) (sim) (r2) (o),
                    ] (back){};
    \end{pgfonlayer}
    \draw[->] (x1)--(x1-|back.west);   
   \draw[->] (x1-|back.west)--(x1-|e1.west);
    \draw[<-] (y1)--(y1-|back.west);
   \draw[<-] (y1-|back.west)--(y1-|e1.west);
    \draw[->] (x2)--(x2-|back.east);   
   \draw[->] (x2-|back.east)--(x2-|e2.east);
    \draw[<-] (y2)--(y2-|back.east);
   \draw[<-] (y2-|back.east)--(y2-|e2.east);
\end{tikzpicture}
\caption{Channel simulator model: collective forward and backward rates satisfy $nR_{12}\ge \sum_{i:odd}H(C_i)$ and $nR_{21}\ge \sum_{i:even}H(C_i)$ respectively. }
\vspace{-.5cm}
\end{figure}
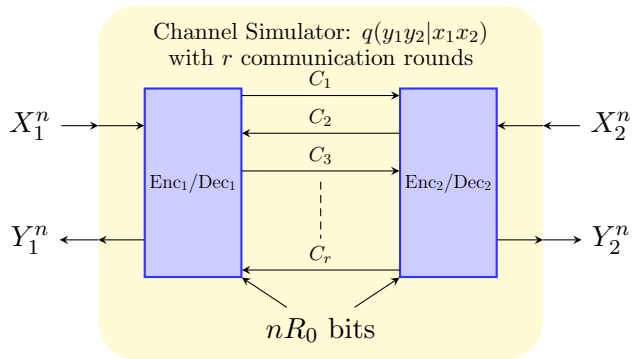
Our result relates to the literature of coordinating distributed controllers to carry out some joint action (see e.g. \cite{coordination,Venkat}) since the generated random variables can be thought of as coordinated actions. Also, our result has implications in quantum information theory. Finding the communication cost of simulating non-local correlations has been subject to many studies where the goal is to simulate an arbitrary bipartite box $p(y_1,y_2|x_1, x_2)$. 
Our result in this paper implies an asymptotic information theoretic characterization of the communication cost (that serves as a lower bound to the one-shot communication complexity formulation whose characterization remains an open problem; e.g. see \cite{Harsha}). As a future work along these lines, it would be interesting to find the entanglement assisted version of our results, similar to the extensions of \cite{cuff} in \cite{Bennett2}.

Lastly we would like to point out that our work falls into the category of strong coordination problems, which has been popularized by Cuff. See \cite{secure}-\cite{farzin13} for some recent works on strong coordination.

This paper is organized as follows: in the next subsection we describe the main proof technique at an intuitive level. In Section \ref{sec:II} we define the problem and in Section \ref{sec:III} we state the main results followed by proofs in Sections \ref{sec:IV} and \ref{sec:V}.
\par\textbf{Notation}: In this paper, 
we use $X_{\ms}$ to denote $(X_j:j\in\ms)$.
we use $p^U_{\ma}$ to denote the uniform distribution over the set $\ma$ and $p(x^n)$ to denote the i.i.d. pmf $\prod_{i=1}^np(x_i)$, unless otherwise stated. The total variation between two pmf's $p$ and $q$ on the same alphabet $\mx$ , is denoted by $\tv{p(x)-q(x)}$. 

\begin{remark} Similar to \cite{cuff-trans} in this work we frequently use the concept of \emph{random} pmfs, which we denote by capital letters (e.g. $P_X$). For any countable set $\mx$ let $\Delta^{\mx}$ be the probability simplex for distributions on $\mx$. A random pmf $P_X$ is a probability distribution over $\Delta^{\mx}$. In other words, if we use $\Omega$ to denote the sample space, the mapping $\omega\in \Omega \mapsto P_X(x;\omega)$ is a random variable for all $x\in\mx$ such that $P_X(x;\omega)\geq 0$ and $\sum_{x}P_X(x;\omega)=1$ for all $\omega$. Thus, $\omega\mapsto P_X(\cdot;\omega)$ is a vector of random variables, which we denote by $P_X$. We  define $P_{X,Y}$ on product set $\mx\times\my$ in a similar way. We note that we can continue to use the law of total probability with random pmfs (e.g. to write $P_X(x)=\sum_{y}P_{XY}(x,y)$ meaning that $P_X(x;\omega)=\sum_yP_{XY}(x,y;\omega)$ for all $\omega$) and the conditional probability pmfs (e.g. to write $P_{Y|X}(y|x)=\frac{P_{XY}(x,y)}{P_X(x)}$ meaning that $P_{Y|X}(y|x;\omega)=\frac{P_{XY}(x,y;\omega)}{P_X(x;\omega)}$ for all $\omega$).
\end{remark}

\section{Problem Statement}
\label{sec:II}
Two terminals observe i.i.d.\ copies of sources $X_1, X_2$ (taking values in finite sets $\mx_1$ and $\mx_2$ and having a joint pmf $q(x_1, x_2)$) respectively. A random variable $\omega$ which is independent of $X_{[1:2]}^n=X_1^nX_2^n$ and is uniformly distributed over $[1:2^{nR_0}]$ represents the \emph{common randomness} provided to the terminals. Given an arbitrary $r\in\mathbb{N}$, an
$(n, R_0, R_{12}, R_{21})$ channel simulation code for simulating a channel
with $r$ \emph{interactive rounds of communications}, consists of
\begin{itemize}
\item a set of $r$ randomized encodings specified with the conditional pmf's $\tilde{p}^{\mathsf{enc}_1}(c_i|c_{[1:i-1]}x_1^n\omega)$ for odd numbers $i\in[1:r]$ and $\tilde{p}^{\mathsf{enc}_2}(c_i|c_{[1:i-1]}x_2^n\omega)$ for even numbers $i\in[1:r]$, where $C_i$ denotes the communication of the $i$-th round,
\item  two randomized decoders $\tilde{p}^{\mathsf{dec}_1}(y_1^n|c_{[1:r]}x_1^n\omega)$ and $\tilde{p}^{\mathsf{dec}_2}(y_2^n|c_{[1:r]}x_2^n\omega)$,
\end{itemize}
such that
\begin{align*}\frac{1}{n}\sum_{i:odd}H(C_i)\le R_{12},~~~\frac{1}{n}\sum_{i:even} H(C_i)\le R_{21}.\end{align*}
\begin{definition}
Given a channel with transition probability $q(y_{[1:2]}|x_{[1:2]})$, a rate tuple $(R_0,R_{12},R_{21})$ is said to be achievable if there exists a sequence of $(n,R_0,R_{12},R_{21})$ channel simulation codes, such that
the total variation between the probability $\tilde{p}(y_{[1:2]}^n,x_{[1:2]}^n)$ induced by the code and the i.i.d.\ repetitions of the desired pmf $q(y_{[1:2]}|x_{[1:2]})q(x_{[1:2]})$ vanishes as $n$ goes to infinity, that is
\be
\label{eq:t}
\lim_{n\rightarrow\infty}\tv{\tilde{p}(y_{[1:2]}^nx_{[1:2]}^n)-\prod_{i=1}^n q(y_{[1:2],i}x_{[1:2],i})}=0.
\ee
\end{definition}
\begin{definition}
The simulation rate region is the closure of all the achievable rate tuples $(R_0,R_{12},R_{21})$.
\end{definition}
\begin{remark}
In the special case $r=1$, $\my_1=\mx_2=\emptyset$, our problem reduces to the one considered by Cuff in \cite{cuff-trans}.
\end{remark}
\begin{remark}
Observe that if $Y_1=f_1(X_{[1:2]})$ and $Y_2=f_2(X_{[1:2]})$ are deterministic functions, the total variation constraint of eq. \eqref{eq:t} reduces to
\begin{align*}\lim_{n\rightarrow \infty}\tilde{p}\left(Y_1^n=f_1(X_{[1:2]}^n), Y_2^n=f_2(X_{[1:2]}^n)\right) =1.\end{align*}
Thus our problem reduces to the problem of \emph{interactive function computation} considered in \cite{ma}.
\end{remark}
\begin{definition}[Empirical coordination \cite{coordination}]\label{rem}
Assume that instead of simulating the channel $q(y_{[1:2]}|x_{[1:2]})$, the demand is to find encoders and decoders such that the output sequences $Y_{[1:2]}^n$ are jointly typical with the inputs $X_{[1:2]}^n$, with high probability. In this case, condition \eqref{eq:t} should be replaced by the following condition:
\be
\label{eq:em}
\lim_{n\rightarrow\infty}\tilde{p}\left(\tv{\tilde{\mathbf{p}}_{X_{[1:2]}^nY_{[1:2]}^n}-q_{X_{[1:2]}Y_{[1:2]}}}>\epsilon\right)=0,\ee
where $\tilde{\mathbf{p}}_{X_{[1:2]}^nY_{[1:2]}^n}$ is the empirical distribution of the
pair $(X_{[1:2]}^n,Y_{[1:2]}^n)$ induced by the chosen code.
\end{definition}
\begin{remark}\label{re:cor}
 It can be shown that if a sequence of codes satisfies the channel simulation condition \eqref{eq:t}, then it also satisfies the empirical coordination constraint \eqref{eq:em}. On the other hand it was shown in \cite[Theorem 2]{coordination} that the empirical rate region does not depend on the amount of \emph{common randomness}, that is, if $(R_0,R_{12},R_{21})$ is achievable for empirical coordination,  then $(0,R_{12},R_{21})$ is also achievable.  These two facts imply that the achievability of a pair of $(R_{12},R_{21})$  for empirical coordination can be proved indirectly through the achievability proof for channel simulation in the presence of an unlimited common randomness. In \cite{coordination}, it was conjectured that this relation is two-sided, i.e. the rate regions for empirical coordination and channel simulation with unlimited common randomness are equal.
\end{remark}
\section{Main Results}
\label{sec:III}
\begin{theorem}[Channel Simulation] \label{thm:m}
The simulation rate region is the set $\ms(r)$ of all non-negative rate tuples $(R_0,R_{12},R_{21})$, for which there exists $p(f_1,\cdots,f_r,x_{[1:2]},y_{[1:2]})\in T(r)$ such that
\begin{align}
R_{12}&\ge I(X_1;\F|X_2),\n
R_{21}&\ge I(X_2;\F|X_1),\n
R_0+R_{12}&\ge I(X_1;\F|X_2)+I(F_1;Y_{[1:2]}|X_{[1:2]}),\n
R_0+R_{12}+R_{21}&\ge I(X_1;\F|X_2)+I(X_2;\F|X_1) +I(\F;Y_{[1:2]}|X_{[1:2]}),
\end{align}
where $T(r)$ is the set of $p(f_1,\cdots,f_r,x_{[1:2]},y_{[1:2]})$ satisfying
\begin{align}
X_{[1:2]}&,Y_{[1:2]}\sim q(x_{[1:2]})q(y_{[1:2]}|x_{[1:2]}),\n
F_i-&F_{[1:i-1]}X_1-X_2, \ \mbox{if $i$ is odd,} \n
F_i-&F_{[1:i-1]}X_2-X_1, \ \mbox{if $i$ is even,} \n
&Y_1-\F X_1-X_2Y_2,\n
&Y_2-\F X_2-X_1Y_1,\n
|\mf_1|&\le|\mx_1||\mx_2||\my_1||\my_2|+3,\n
\forall i>1:
|\mf_i|&\le|\mx_1||\mx_2||\my_1||\my_2|\prod_{j=1}^{i-1}|\mf_j|+2.
\end{align}
\end{theorem}
\begin{corollary}[Interactive function computation \cite{ma}]
Assume that the desired channel is deterministic, that is, $Y_1=f_1(X_{[1:2]})$ and $Y_2=f_2(X_{[1:2]})$. Setting $R_0=0$ in Theorem \ref{thm:m} gives the following full characterization of the rate region of reliable interactive computation,
\bes\begin{split}
\mr(r)=\{\exists \F:&R_{12}\ge I(X_1;\F|X_2)\\
&R_{21}\ge I(X_2;\F|X_1)\\
&F_i-F_{[1:i-1]}X_1-X_2, \ \mbox{if $i$ is odd,} \\
&F_i-F_{[1:i-1]}X_2-X_1, \ \mbox{if $i$ is even,} \\
&H(Y_2|\F X_2)=H(Y_1|\F X_1)=0\}.
\end{split}\ees
\end{corollary}
\begin{theorem}[Empirical coordination]
\label{thm:emp}
The empirical coordination rate region is the set of all non-negative rate pairs $(R_{12},R_{21})$, for which there exists $p(f_1,\cdots,f_r,x_{[1:2]},y_{[1:2]})\in T(r)$ (defined in Theorem \ref{thm:m}) such that
\begin{align}
R_{12}&\ge I(X_1;\F|X_2),\n
R_{21}&\ge I(X_2;\F|X_1).
\end{align}
\end{theorem}
The achievability part of this theorem comes from setting $R_0=\infty$ in Theorem \ref{thm:m}. The converse is relegated to Appendix \ref{apx:conv}.  
\begin{remark}
Interactive empirical coordination is related to the problem of interactive lossy source coding solved by Kaspi in \cite{kaspi}. The above theorem in conjunction with \cite[Theorem 9]{cuff:thesis} provides an alternative proof for that result.
\end{remark}
\subsection{Non-Symmetry of the Simulation Region}
One may expect that the simulation region should be symmetric. However the region $\ms(r)$ is not symmetric; in particular there is an interesting inequality on $R_0+R_{12}$ which is not symmetric. In this regard, the following observations are useful:
\begin{enumerate}
\item We are finding the region for a finite $r$ rounds of communication.  Since $r$ is fixed and the \underline{first party} starts the communication, there will be a non-symmetry.
\item The region would have been symmetric if the region was for infinite rounds of communication (i.e. $r\rightarrow\infty$). We prove this by showing that the constraint on $R_0+R_{12}$ can be relaxed from the definition of $\ms(r)$, when we want to compute $\ms(\infty)=\bigcup_{r\ge1}\ms(r)$. Let $\ms'(r)$ be the rate region obtained from $\ms(r)$ by relaxing the constraint on $R_0+R_{12}$. We show that $\ms'(r)\subseteq\ms(r+2)$. Take a point $(R_0,R_{12},R_{21})\in\ms'(r)$. Let $F_{[1:r]}$ be the corresponding  random variables for this point.  We find $F'_{[1:r+2]}$ to reach $(R_0, R_{12}, R_{21})$ as a point in $\ms(r+2)$. Define $F'_1=\emptyset, F'_2=\emptyset$ and $F'_i=F_{i-2}$ for $i>2$. Then writing the constraint corresponding to $F'_{[1:r+2]}$ and removing a redundant equation (the one on $R_0+R_{12}$) gives us what we need. 

\item Communication can itself be used to establish common randomness. For instance the first party can allocate parts of its first message to create common randomness. This implies that if the point $(R_0, R_1, R_2)$ is in $\mathcal{S}( r)$, so is $(R_0-\alpha, R_1+\alpha, R_2)$ in $\mathcal{S}( r)$ for positive $\alpha\leq R_0$. The given region has this property. 

On the other hand, if the second party wants to use its communication to generate common randomness, the first party who is initiating the communication cannot use this common randomness in the first round, which becomes a different setup from the one we are considering here. However, a special use of the communication by the second party to generate common randomness is to have the first party not sending anything in the first round. The second party sends a message of size $\alpha$ to be used as common randomness in the next rounds.  This implies that if the point $(R_0, R_1, R_2)$ is in $\mathcal{S}( r)$, so is $(R_0-\alpha, R_1, R_2+\alpha)$ in $\mathcal{S}( r+2)$ for positive $\alpha\leq R_0$. To show that our region has this property assume that we use $F_{[1:r]}$ to reach the point $(R_0, R_1, R_2)$ in $\mathcal{S}( r)$. We define $F'_{[1:r+2]}$ to reach $(R_0-\alpha, R_1, R_2+\alpha)$ as follows: $F'_1=\emptyset, F'_2=\emptyset$ and $F'_i=F_{i-2}$ for $i>2$. Substituting this and removing a redundant equation (the one on $R_0+R_{12}$) gives us what we need. 

We can also consider the case where the first party does not send anything, and the second party allocates parts of its message in the second round to generating common randomness. When the first party does not talk in the first round, it is as if the role of the second party and the first are S-Witched, but the number of rounds is increased by one. One can verify that the region given in the statement of the theorem still has the expected properties.


\end{enumerate}

\section{Achievability}
\label{sec:IV}
\subsection{Review of the output statistics of random binning}
\subsubsection{Description of the proof technique}
\label{sec:I:II}
In this paper, the achievability part of the proof is based on the technique of ``output statistics of random binning" (OSRB) that has been recently developed in \cite{me}. To explain the technique we begin with describing the resolvability lemma used by Cuff \cite[Lemma 6.1]{cuff}, and originally proved by Wyner. We report this lemma in a slightly different form that suits our purpose. Although we do not use this lemma in this work, since it is very central to the achievability proof of \cite{cuff-trans}, we illustrate how this lemma can be proved using the OSRB approach. 

To discuss the resolvability lemma  \cite[Lemma 6.1]{cuff}, let us fix some $p(x,y)$. Roughly speaking the lemma states that one can find $2^{nR}$ sequences in $\mx^n$, namely $x^n(1), x^n(2),\cdots, x^n({2^{nR}})$, such that if we choose one of these sequences at random and pass it through the DMC channel $p(y|x)$ we get an output sequence that is almost i.i.d. according to $p(y)$, as long as $R>I(X;Y)$. 
We can restate this lemma by letting $M$ to be a random variable whose alphabet is $\mm=[1:2^{nI(X;Y)}]$, and assuming that $X^n(M)$ is transmitted over the DMC channel $q(y|x)$. To prove this lemma in the traditional way one would construct a random codebook parametrized by a random variable $B$. Every choice of $B=b$ corresponds to a particular codebook (particular set of sequences in the $\mx^n$ space). The probability distribution imposed on the $\my^n$ space depends on the value of $B$, which is itself random. Therefore we use the capital letter $P_{Y^n}$ to denote the random p.m.f. induced on $\my^n$, by the random codebook. 
 To show the above lemma one would need to show that the expected value of the total variation distance between the probability measure $y^n\mapsto P(y^n)$ and the i.i.d.\ distribution is small. Therefore there exists $B=b$ where the total variation distance is small. Indeed this is the way that Cuff proves this lemma in \cite[Lemma IV. 1]{cuff-trans}, \cite[Lemma 19]{cuff:thesis}.

To illustrate the proof of this lemma using the OSRB approach, one would need to start from $n$ i.i.d.\ copies of $X^n$ and $Y^n$ from the given $p(x,y)$. Random variables $B$ and $M$ are identified as random binnings of $X^n$ at rates $2^{n\tR}$ and $2^{nR}$ respectively. Note the conceptual change is in starting from the i.i.d.\ distribution and then defining $B$ as a function of $X^n$. It is proved that if $\tR<H(X|Y)$, $B$ is almost independent of $Y^n$. Therefore, for almost any choice of $B=b$, the distribution of $Y^n$ conditioned on $B=b$ is almost i.i.d. On the other hand, if $\tR+R>H(X)$, $X^n$ will be a function of $(M,B)$ with high probability by the Slepian-Wolf. We are interpreting $B$ and $M$ as two messages coming from two encoders both observing $X^n$. These imply that one can find $B=b$ such that the conditional law $y^n\mapsto p(y^n|B=b)$ is close to the i.i.d.\ distribution, and at the same time $X^n$ is almost a function of $M$ conditioned on $B=b$. All the approximations in this intuitive argument can be made accurate. 

The crucial departure from the traditional argument was our treatment of random variable $B$. As discussed in \cite{me}, the randomness in generating a random codebook is generally conceived of a common randomness shared among the terminals in a problem. However, we are changing the order by first generating i.i.d.\ distributions and then treating $B$ as a random binning on this product i.i.d. space. 

\subsubsection{Main tools}
Let $(X_{[1:T]},Y)$ be a discrete memoryless correlated sources (DMCS) distributed according to a joint pmf $p_{X_{[1:T]},Y}$ on finite sets. A distributed  random  binning consists of a set of random mappings $\mb_i: \mx_i^n\rightarrow [1:2^{nR_i}]$, $i\in[1:T]$, in which $\mb_i$ maps each sequence of $\mx_i^n$ uniformly and independently to $[1:2^{nR_i}]$. We denote the random variable $\mb_t(X_t^n)$ by $B_t$. A random distributed  binning induces the following \emph{random pmf} on the set $\mx_{[1:T]}^n\times\my^n\times\prod_{t=1}^T [1:2^{nR_t}]$,
\[
P(x^n_{[1:T]},y^n,b_{[1:T]})=p(x_{[1:T]}^n,y^n)\prod_{t=1}^T\mathbf{1}\{\mb_t(x_t^n)=b_t\}.
\]
\begin{theorem}[\cite{me}]
\label{thm:re}
If for each $\ms\subseteq [1:T]$, the following constraint holds
\be
\sum_{t\in\ms}R_t<H(X_{\ms}|Y),
\ee
then as $n$ goes to infinity, we have
\be
\e\tv{P(y^n,b_{[1:T]})-p(y^n)\prod_{t=1}^T p^U_{[1:2^{nR_t}]}(b_t)}\rightarrow 0.
\ee
\end{theorem}

We now consider another region for which we can approximate a specified pmf. This region is the Slepian-Wolf (S-W) region for reconstructing $X^n_{[1:T]}$ in the presence of $(B_{[1:T]}, Y^n)$ at the decoder. 
 As in the achievability proof of the \cite[Theorem 15.4.1]{cover:book}, we can define a decoder with respect to any fixed distributed binning. We denote the decoder by the random conditional pmf $P^{S-W}(\hat{x}^n_{[1:T]}|y^n,b_{[1:T]})$ (note that since the decoder is a function, this pmf takes only two values, 0 and 1). Now we write the Slepian-Wolf theorem in the following equivalent form. See \cite{me} for details.
\begin{lemma}[\cite{me}]\label{le:S-W}
If for each $\ms\subseteq [1:T]$, the following constraint holds 
\be
\sum_{t\in\ms}R_t>H(X_{\ms}|X_{\ms^c}, Y),
\ee
then as $n$ goes to infinity, we have
\bes
\e\tv{P(x^n_{[1:T]},y^n,\hat{x}^n_{[1:T]})-p(x^n_{[1:T]},y^n)\mathbf{1}\{\hat{x}^n_{[1:T]}=x^n_{[1:T]}\}}\rightarrow 0.
\ees
\end{lemma}  
\begin{definition}\label{def:1}
For any random pmfs $P_X$ and $Q_X$ on $\mx$, we say $P_X\stackrel{\epsilon}{\approx}Q_X$ if $\e\tv{P_X-Q_X}<\epsilon$. Similarly we use $p_X\apx{\epsilon}q_x$ for two (non-random) pmfs to denote the total variation constraint $\tv{p_X-q_X}<\epsilon$.\end{definition}

\begin{lemma}[\cite{me}]\label{le:total}We have
\begin{enumerate}
\item
$\tv{p_Xp_{Y|X}-q_{X}p_{Y|X}}=\tv{p_X-q_X}$\\
$~~~~~~~~~~~\quad\tv{p_X-q_X}\le\tv{p_Xp_{Y|X}-q_{X}q_{Y|X}}$
\item If $p_Xp_{Y|X}\stackrel{\epsilon}{\approx}q_Xq_{Y|X}$, then there exists $x\in\mx$ such that $p_{Y|X=x}\stackrel{2\epsilon}{\approx}q_{Y|X=x}$.
\item If $P_X\stackrel{\epsilon}{\approx}Q_X$ and  $P_XP_{Y|X}\stackrel{\delta}{\approx}P_XQ_{Y|X}$, then $P_{X}P_{Y|X}\stackrel{\epsilon+\delta}{\approx}Q_{X}Q_{Y|X}$.
\end{enumerate}
\end{lemma}
\subsection{Achievability proof of Theorem \ref{thm:m}}\label{sec:ach}
We use a combination of the Slepian-Wolf theorem and Theorem \ref{thm:re} to prove Theorem \ref{thm:m}. 
\par 
The proof is divided into three parts. In the first part we introduce two protocols each of which induces a pmf on a certain set of r.v.'s. The first protocol has the desired i.i.d.\ property on $(X_{[1:2]}^n,Y_{[1:2]}^n)$ but leads to no concrete coding algorithm. However the second protocol is suitable for construction of a code, with one exception: the second protocol is assisted with an extra common randomness that does not really exist in the model. In the second part we find constraints on $R_0, R_{12}, R_{21}$ implying that these two induced distributions are almost identical. In the third part of the proof, we eliminate the extra common randomness given to the second protocol without disturbing the pmf induced on the desired random variables ($X_{[1:2]}^n,Y_{[1:2]}^n$) significantly. This makes the second protocol useful for code construction.

\emph{Part (1) of the proof:} 
Take an arbitrary $p(f_{[1:r]}, x_{[1:2]}, y_{[1:2]})\in T(r)$. Let $R_i$ be the rate of the communication at round $i$. Thus we have 
\be R_{12}=\sum_{i:odd} R_i,~~~~~R_{21}=\sum_{i:even}R_i.\label{eq:rate-split}\ee
 We define two protocols each of which induces a joint distribution on random variables that are defined during the protocol. 

\emph{Protocol A. } We begin by describing a random binning strategy that we will use when defining Protocol A. 

\underline{Random Binning}: Let $(\F^n, X_{[1:2]}^n, Y_{[1:2]}^n)$ be i.i.d. and distributed according to $p(f_{[1:r]}, x_{[1:2]}, y_{[1:2]})$. Since $p(f_{[1:r]}, x_{[1:2]}, y_{[1:2]})\in T(r)$, it factors as
\bes
\label{eq:1}\begin{split}
p(x_{[1:2]}^n)\left[\prod_{i=1}^r{p}(f_i^n|f_{[1:i-1]}^nx_{(i)_2}^n)\right]{p}(y_1^n|f_{[1:r]}^nx_1^n) p(y_2^n|f_{[1:r]}^nx_2^n).
\end{split}\ees

Consider the following random binning:
\begin{itemize}
\item To each sequence $f_1^n$, assign uniformly and independently three bin indices $b_1\in[1:2^{n\tR_1}]$, $k_1\in[1:2^{nR_1}]$ and $\omega\in[1:2^{nR_0}]$. 
\item For $i\in[2:r]$, to each sequence $(f_1^n,\cdots,f_i^n)$, assign uniformly and independently two bin indices $b_i\in[1:2^{n\tR_i}]$ and $k_i\in[1:2^{nR_i}]$. 
\end{itemize}
Furthermore, for $i\in[1:r]$, we consider the Slepian-Wolf decoder for recovering $\hat{f}_{i}^n$ from $(f_{[1:i-1]}^n,b_i,k_i,\omega,x_{(i+1)_2}^n)$ and denote it by $P^{S-W}(\hat{f}_{i}^n|b_i,k_i,\omega,f_{[1:i-1]}^n,x_{(i+1)_2}^n)$.  Note that we denote the estimates of $f_i^n$ by $\hat{f}_{i}^n$. The rate constraints for the success of these decoders will be imposed later, although these decoders can be conceived even when there is no guarantee of success.  

We  define $\hat{f}_{i,\mathsf{T}_1}^n$ for terminal 1, i.e. $\mathsf{T}_1$, and $\hat{f}_{i,\mathsf{T}_2}^n$ for terminal 2, i.e. $\mathsf{T}_2$ as follows:
\[\hat{f}_{i,\mathsf{T}_1}^n=\left\{\begin{array}{lc}f_i^n & \mbox{for odd $i$}\\
                                                                 \hat{f}_i^n& \mbox{for even $i$,}
                                             \end{array}  \right.\]
and                                                               
\[\hat{f}_{i,\mathsf{T}_2}^n=\left\{\begin{array}{lc}f_i^n & \mbox{for even $i$}\\
                                                                 \hat{f}_i^n& \mbox{for odd $i$.}
                                             \end{array}  \right.\]
The random
pmf induced by the random binning, denoted by $P$, can be expressed as follows:
\begin{align}
P(x_{[1:2]}^n,&f_{[1:r]}^n,b_{[1:r]}, k_{[1:r]}, \omega, y_{[1:2]}^n, \hat{f}_{[1:r],\ta}^n, \hat{f}_{[1:r],\tb}^n)\n&=
p(x_{[1:2]}^n)\left[\prod_{i=1}^r{p}(f_i^n|f_{[1:i-1]}^n, x_{(i)_2}^n)P(b_i,k_i,\omega_i|f_i^n,f_{[1:i-1]}^n)\right.\n&\qquad\qquad\qquad\left. P^{S-W}(\hat{f}_{i,\t{i+1}}^n|b_i,k_i,\omega_i,f_{[1:i-1]}^n,x^n_{(i+1)_2})\mathbf{1}\{\hat{f}^n_{i,\t{i}}=f^n_{i}\}\right] {p}(y_1^n|f_{[1:r]}^nx_1^n)p(y_2^n|f_{[1:r]}^nx_2^n)\n &= p(x_{[1:2]}^n)
\left[\prod_{i=1}^rP(b_i,\omega_i|f_{[1:i-1]}^n, x_{(i)_2}^n)P(f_i^n,k_i|b_i,\omega_i,f_{[1:i-1]}^n, x_{(i)_2}^n)\right.\n&\qquad\qquad\qquad\left. P^{S-W}(\hat{f}_{i,\t{i+1}}^n|b_i,k_i,\omega_i,f_{[1:i-1]}^n,x^n_{(i+1)_2})\mathbf{1}\{\hat{f}^n_{i,\t{i}}=f^n_{i}\}\right] {p}(y_1^n|f_{[1:r]}^nx_1^n)p(y_2^n|f_{[1:r]}^nx_2^n),\label{eq:3}
\end{align}
where $(i)_2:=i \mod 2$ and $\omega_1=\omega$, and $\omega_i$ is a constant variable for $i\ge 2$.

\emph{Protocol B.} Given some $p(f_{[1:r]}, x_{[1:2]}, y_{[1:2]})\in T(r)$, we define Protocol B as follows: 
In this protocol we assume that the terminals have access to the
shared randomness $B_{[1:r]}$ where $B_{[1:r]}$ are mutually independent r.v.'s and uniformly distributed on $\prod_{t=1}^r[1:2^{n\tR_t}]$. R.v. $\omega$ is also used for the common randomness (it is independent of $B_{[1:r]}$). The shared randomness $B_{[1:r]}$ does not really exist in the real model, and we will get rid of it later. However $\omega$ is the actual common randomness shared between the two terminals in the model. Random variable $K_i$ is used for the communication at round $i$. 
Then, the protocol proceeds as follows,
\begin{itemize}
\item In the first round, knowing $(b_1,\omega,x_1^n)$, terminal $1$ generates a sequence $f_1^n$ according to $P(f_1^n|b_1,\omega,x_1^n)$ of protocol A, and sends the bin index $k_1(f_1^n)$ of protocol A to the terminal $2$. At the end of the first round, terminal $2$ having $(b_1,\omega,k_1,x_2^n)$, uses the Slepian-Wolf decoder $P^{S-W}(\hat{f}_{1}^n|b_1,k_1,\omega,x_{2}^n)$ of protocol A to obtain an estimate of $f_1^n$. We use $\hat{f}_{1,\mathsf{T}_2}^n$ to denote this estimate of $f_1^n$ by the second terminal $\mathsf{T}_2$. Since the first terminal knows $f_1^n$ we set $\hat{f}_{1,\mathsf{T}_1}^n=f_1^n$ to be the estimate of $f_1^n$ by the first terminal $\mathsf{T}_1$.
\item In the second round, knowing $(b_2,x_2^n,\hat{f}_{1,\tb}^n)$, terminal $2$ generates a sequence $f_2^n$ according to $\\P_{F_2^n|B_2X_2^nF_1^n}(f_2^n|b_2,x_2^n,\hat{f}_{1,\tb}^n)$ of protocol A and sends the bin index $k_2(\hat{f}_{1,\tb}^n,\hat{f}_{2,\tb}^n)$ of Protocol A to the terminal $1$. At the end of the second round, terminal $1$ having $(b_2,k_2,x_1^n,\hat{f}_{1,\ta}^n)$, uses the $\\P^{S-W}_{\hat{F}_2^n|B_2,K_2,\omega,F_1^n,X_1^n}(\hat{f}_{2}^n|b_2,k_2,\omega,\hat{f}_{1,\mathsf{T}_1}^n,x_{1}^n)$ defined above to recover $\hat{f}_{2}^n$. We omit subscripts  from the pmfs when they are clear from the context.
\item This procedure is repeated interactively for $i\in[3:r]$. Thus, at the end of the round $r$, the first terminal has $\hat{f}_{[1:r],\ta}^n$ and the second terminal has $\hat{f}_{[1:r],\tb}^n$. 
\item The first terminal uses the conditional distribution $p(y_1|x_1,f_{[1:r]})$ that we started with at the beginning to create $y_1^n$ from the conditional distribution $p(y_1^n|x_1^n,\hat{f}^n_{[1:r],\mathsf{T}_1})$ and the second terminal uses the conditional distribution $p(y_2|x_2,f_{[1:r]})$ to create $y_2^n$ from $p(y_2^n|x_2^n,\hat{f}^n_{[1:r],\mathsf{T}_2})$.
\end{itemize}

The random pmf induced by the protocol, denoted by $\widehat{P}$, factors as
\begin{align}
\widehat{P}(x_{[1:2]}^n,&f_{[1:r]}^n,b_{[1:r]}, k_{[1:r]}, \omega, y_{[1:2]}^n, \hat{f}_{[1:r],\ta}^n, \hat{f}_{[1:r],\tb}^n)\n&=
p(x_{[1:2]}^n)p^U(\omega)p^U(b_{[1:r]})
\left[\prod_{i=1}^rP(f_i^n,k_i|b_i,\omega_i,\hat{f}_{[1:i-1],\t{i}}^n, x_{(i)_2}^n)\right.\n&~~\left.P^{S-W}(\hat{f}_{i,\t{i+1}}^n|b_i,k_i,\omega_i,\hat{f}_{[1:i-1],\t{i+1}}^n,x^n_{(i+1)_2}) \mathbf{1}\{\hat{f}^n_{i,\t{i}}=f^n_{i}\}\right] {p}(y_1^n|\hat{f}_{[1:r],\ta}^nx_1^n)p(y_2^n|\hat{f}_{[1:r],\tb}^nx_2^n)\n&=
p(x_{[1:2]}^n)
\left[\prod_{i=1}^rp^U(\omega_i)p^U(b_{i}) P(f_i^n,k_i|b_i,\omega_i,\hat{f}_{[1:i-1],\t{i}}^n, x_{(i)_2}^n)\right.\n&\left.\qquad\qquad~~~\qquad P^{S-W}(\hat{f}_{i,\t{i+1}}^n|b_i,k_i,\omega_i,\hat{f}_{[1:i-1],\t{i+1}}^n,x^n_{(i+1)_2}) \mathbf{1}\{\hat{f}^n_{i,\t{i}}=f^n_{i}\}\right]\n&\qquad\qquad\qquad\qquad\qquad~~~~~~~~~ {p}(y_1^n|\hat{f}_{[1:r],\ta}^nx_1^n)p(y_2^n|\hat{f}_{[1:r],\tb}^nx_2^n).
\label{eq:4}
\end{align}
where $\omega_1=\omega$, and $\omega_i$ is a constant variable for $i\ge 2$.

\emph{Part (2) of the proof: Sufficient conditions that make the induced pmfs approximately the same}:
We need to find conditions that imply
\begin{align*}
P(x_{[1:2]}^n,f_{[1:r]}^n,b_{[1:r]}, k_{[1:r]}, \omega, y_{[1:2]}^n, \hat{f}_{[1:r],\ta}^n, \hat{f}_{[1:r],\tb}^n)\apx{\epsilon_n}
\widehat{P}(x_{[1:2]}^n,f_{[1:r]}^n,b_{[1:r]}, k_{[1:r]}, \omega, y_{[1:2]}^n, \hat{f}_{[1:r],\ta}^n, \hat{f}_{[1:r],\tb}^n)
\end{align*}
for some $\epsilon_n$ converging to zero as $n\rightarrow\infty$.
We begin by proving that
\begin{align*}
P(x_{[1:2]}^n,f_{[1:r]}^n,b_{[1:r]}, k_{[1:r]}, \omega,\hat{f}_{[1:r],\ta}^n, \hat{f}_{[1:r],\tb}^n)\apx{\epsilon_n}
\widehat{P}(x_{[1:2]}^n,f_{[1:r]}^n,b_{[1:r]}, k_{[1:r]}, \omega, \hat{f}_{[1:r],\ta}^n, \hat{f}_{[1:r],\tb}^n),
\end{align*}
where we have dropped $y_{[1:2]}^n$ from both sides. We will add $y_{[1:2]}^n$ to the equation later.


To find the constraints that imply that the pmf $\widehat{P}$ is close to the pmf $P$ in total variation distance, 
we start with $P$ and make it close to $\widehat{P}$ in a few steps. For any $j\in[0:r]$ we inductively find constraints that imply
\begin{align}
P(x_{[1:2]}^n,f_{[1:j]}^n,b_{[1:j]}, k_{[1:j]}, \omega_j, \hat{f}_{[1:j],\ta}^n, \hat{f}_{[1:j],\tb}^n)\apx{\epsilon_n}
\widehat{P}(x_{[1:2]}^n,f_{[1:j]}^n,b_{[1:j]}, k_{[1:j]}, \omega_j, \hat{f}_{[1:j],\ta}^n, \hat{f}_{[1:j],\tb}^n).\label{eqn:apx-A}
\end{align}
for some $\epsilon_n$ converging to zero as $n\rightarrow\infty$, where $\omega_1=\omega$, and $\omega_j$ is a constant variable for $j\ge 2$.
For $j=0$ this is trivial since it reduces to $P(x_{[1:2]}^n)=p(x_{[1:2]}^n)=\widehat{P}(x_{[1:2]}^n)$. We show in Appendix \ref{apx:a} that the constraints sufficient to guarantee the statement for $j$ given that it holds for $j-1$ are as follows:
\begin{enumerate}
\item \emph{Reliability of S-W decoders:} For $j=1$ the S-W decoding is reliable if,
\be\label{eq:c1}
R_1+R_0+\tR_1\ge H(F_1|X_2).
\ee
For $j\geq 2$ the S-W decoding is reliable if,
\be\label{eq:c2}
\forall i\in[2:r]:\ R_i+\tR_i\ge H(F_i|X_{(i+1)_2}F_{[1:i-1]}).
\ee
\item \emph{Other constraints:} For $j=1$ we have the constraint
\be\label{eq:c3V1}
\begin{split}
R_0+\tR_1&< H(F_1|X_1).
\end{split}
\ee
For $j\geq 2$ we have the constraints
\be\label{eq:c3V2}
\begin{split}
\tR_i&<H(F_i|X_{(i)_2}F_{[1:i-1]}),\ \mbox{for $i=2,\cdots,r$}.
\end{split}
\ee
\end{enumerate}

The details can be found in Appendix \ref{apx:a}, but a brief description is in order. Comparing equations \eqref{eq:3} and \eqref{eq:4} we see that most of the terms are the same if we assume that the Slepian-Wolf decoders succeed with probability one. Constraints \eqref{eq:c1} and \eqref{eq:c2} guarantee the success of the Slepian-Wolf decoders with high probability. Now we assume that the Slepian-Wolf decoders succeed with probability one, that is $\hat{F}_{i,\t{i}}=F_i$ for each $i$. To make the two pmfs close we need to guarantee that $P(b_i,\omega_i|{f}_{[1:i-1]}^n, x_{(i)_2}^n)\apx{} p^U(\omega_i)p^U(b_{i})$. In other words we need constraints ensuring the uniformity of $(b_i,\omega_i)$ and its independence of $({f}_{[1:i-1]}^n, x_{(i)_2}^n)$. These constraints are given in equations \eqref{eq:c3V1} and \eqref{eq:c3V2}, and are obtained using Theorem \ref{thm:re}.


Therefore equations \eqref{eq:3}, \eqref{eq:4}, \eqref{eq:c3V1} and \eqref{eq:c3V2} imply that
\begin{align*}
P(x_{[1:2]}^n,f_{[1:r]}^n,b_{[1:r]}, k_{[1:r]}, \omega,\hat{f}_{[1:r],\ta}^n, \hat{f}_{[1:r],\tb}^n)\apx{\epsilon_n}
\widehat{P}(x_{[1:2]}^n,f_{[1:r]}^n,b_{[1:r]}, k_{[1:r]}, \omega, \hat{f}_{[1:r],\ta}^n, \hat{f}_{[1:r],\tb}^n).
\end{align*}
In Appendix \ref{apx:b} we show that the equations \eqref{eq:3}, \eqref{eq:4} imply that 
\begin{align}
P(x_{[1:2]}^n,f_{[1:r]}^n,b_{[1:r]}, k_{[1:r]}, \omega,y_{[1:2]}^n,\hat{f}_{[1:r],\ta}^n, \hat{f}_{[1:r],\tb}^n)\apx{\epsilon_n}
\widehat{P}(x_{[1:2]}^n,f_{[1:r]}^n,b_{[1:r]}, k_{[1:r]}, \omega,y_{[1:2]}^n,\ \hat{f}_{[1:r],\ta}^n, \hat{f}_{[1:r],\tb}^n).\label{eqn:apx-B}
\end{align}

Using part one of Lemma \ref{le:total} we can deduce the same approximation over the marginals
\begin{align}\label{eq:eq10}
\widehat{P}(&b_{[1:r]},x_{[1:2]}^n, y_{[1:2]}^n)\apx{\epsilon_n}P(b_{[1:r]},x_{[1:2]}^n,y_{[1:2]}^n)\end{align}
for some $\epsilon_n$ converging to zero as $n\rightarrow\infty$. In particular, the marginal pmf of $(X_{[1:2]}^n, Y_{[1:2]}^n)$ of the RHS of this expression is equal to $p(x_{[1:2]}^n,y_{[1:2]}^n)$ which is the desired pmf.

\emph{Part (3) of the proof: Eliminating the shared randomness:} 

  In the protocol we assumed that the terminals have access to shared randomness $B_{[1:r]}$ which is not present in the model (note that $\omega$ is the real common randomness shared between the two terminals in the model). To get rid of the shared randomness $B_{[1:r]}$, we would like to condition on a particular instance of $B_{[1:r]}=b_{[1:r]}$. However, conditioning on $b_{[1:r]}$ may change the marginal pmf of $X_{[1:2]}^n, Y_{[1:2]}^n$ on the LHS of \eqref{eq:eq10}. Thus, we want to impose certain constraints on the size of the bins to guarantee that the marginal pmfs do not change. In other words the induced pmf $\widehat{P}(x_{[1:2]}^n,y_{[1:2]}^n)$ changes to the conditional pmf $\widehat{P}(x_{[1:2]}^n,y_{[1:2]}^n|b_{[1:r]})$. But if $B_{[1:r]}$ is independent of $(X_{[1:2]}^n, Y_{[1:2]}^n)$, then the conditional pmf $\widehat{P}(x_{[1:2]}^n,y_{[1:2]}^n|b_{[1:r]})$ is also close to the desired distribution. Therefore we can assume that the terminals agree on an instance $b_{[1:r]}$ of $B_{[1:r]}$ and run protocol B. More precisely, to obtain the independence, we use Theorem \ref{thm:re} where we substitute $T=r$, $X_i=F_{[1:i]}$ and $Y=X_{[1:2]}Y_{[1:2]}$ to get the following sufficient condition for the pmf of protocol A:\footnote{{Here we only write the constraints corresponding to the subsets of $[1:r]$ of the form $[1:i], 1\le i\le r$ and omit the others, because the unwritten constraints are redundant. This is because the random variables $X_i=F_{[1:i]}$ are \emph{nested} r.v.'s. Each subset of $[1:r]$ can be written as $\ms=\{ m_1,m_2,\cdots,m_k\}$ where $\{m_j\}_{j=1}^k$ is an increasing sequence. In this case $X_{\ms}=F_{[1:m_k]}=X_{[1:m_k]}$ and the corresponding constraint is implied by the constraint corresponded to $[1:m_k]$.}}
\be\label{eq:c44}
\forall i\in[1:r],~~~\sum_{t=1}^i \tR_t<H(F_{[1:i]}|X_{[1:2]}Y_{[1:2]}).
\ee
This implies that
\begin{align}\label{eq:eq11}
P(b_{[1:r]},x_{[1:2]}^n,y_{[1:2]}^n)\apx{\delta_n}p^U(b_{[1:r]})p(x_{[1:2]}^n,y_{[1:2]}^n).\end{align}
Equations \eqref{eq:eq10} and \eqref{eq:eq11} in conjunction with the third part of Lemma \ref{le:total} imply that 
\begin{align}\label{eq:eq12}
\widehat{P}(&b_{[1:r]},x_{[1:2]}^n, y_{[1:2]}^n)\apx{\epsilon_n+\delta_n}p^U(b_{[1:r]})p(x_{[1:2]}^n,y_{[1:2]}^n).\end{align}
 Using Definition \ref{def:1}, equation \eqref{eq:eq12} guarantees existence of  a fixed binning with the corresponding pmf $p$ such that if we replace $P$ with $p$ in \eqref{eq:3} and denote the resulting pmf with $\hat{p}$. This would then imply that
 \begin{align*}
\hat{p}(&b_{[1:r]},x_{[1:2]}^n, y_{[1:2]}^n)\apx{\epsilon_n+\delta_n}p^U(b_{[1:r]})p(x_{[1:2]}^n,y_{[1:2]}^n).
 \end{align*}
 Now, the second part of Lemma \ref{le:total} shows that there exists an instance $b_{[1:r]}$ such that  
 \begin{align*}
\hat{p}(&x_{[1:2]}^n, y_{[1:2]}^n|b_{[1:r]})\apx{2\epsilon_n+2\delta_n}p(x_{[1:2]}^n,y_{[1:2]}^n).
\end{align*}

We have found all the necessary constraints on the size of the bins for protocol to work. Finally, eliminating $(\tR_1,\cdots,\tR_r)$ and $(R_1,\cdots,R_r)$ from the inequalities \eqref{eq:rate-split},\eqref{eq:c1}-\eqref{eq:c3V2} and \eqref{eq:c44} gives rise to the constraints given in the statement of the problem. This is done in Appendix \ref{apx:elimination}. This completes the proof of the achievability of Theorem \ref{thm:m}. 
    

\section{Converse}
\label{sec:V}
We follow the steps used in \cite{cuff-trans} to prove the converse of Theorem \ref{thm:m}. First for any $\epsilon>0$, we find a set $\ms_{\epsilon}(r)$ to be constituted as an outer region for channel simulation region. Then we discuss the continuity of $\ms_{\epsilon}(r)$ at $\epsilon=0$. In particular we show that $\ms(r)=\bigcap_{\epsilon>0}\ms_{\epsilon}(r)$. 

Let $(R_0,R_{12},R_{21})$ be an achievable rate tuple for $r$ rounds of communications. Then, for any $\epsilon<\frac{1}{2}$, there exists a simulation code of length $n$ such that the total variation between the induced pmf $\tilde{p}(y_{[1:2]}^n,x_{[1:2]}^n)$ and the $n$ i.i.d.\ repetitions  of the desired pmf $q(x,y)$ is less than $\epsilon$. 
\subsection{Mutual information bounds}
The following lemmas which are consequences of a generalized version of Lemma 2.7 of \cite{csiszar}, will be useful throughout the proof of the converse. The proofs are provided in the Appendix \ref{apx:d}.
\begin{lemma}\label{le:c1} For any discrete random variables $W^n$ and $Z$ whose joint pmf satisfies
\[
\tv{p(w^n,z)-p(z)\prod_{q=1}^n \widehat{p}_q(w_q|z)}<\epsilon<\frac{1}{2},
\]
for some $\widehat{p}_q(w|z)$, we have
\[\sum_{q=1}^n I(W_q;W^{q-1}|Z)\le 2n\left(\epsilon\log|\mw|+h_b(\epsilon)\right)
,\]
where $h_b(.)$ is the binary entropy function.
\end{lemma}

\begin{lemma}\label{le:c2}
Take an arbitrary i.i.d. sequence $X^n$ distributed according to $p(x)$ and a conditional pmf $p(y^n|x^n)$ which is not necessarily i.i.d. If there exists a conditional pmf $\widehat{p}(y|x)$ such that
\[\tv{p(y^n|x^n)\prod_{q=1}^{n} p(x_q)-\prod_{q=1}^n p(x_q)\widehat{p}(y_q|x_q)}<\epsilon<\frac{1}{2},\]
then 
\be
\forall q\in[1:n]:~~~~~I(X_{[\sim q]};Y_q|X_q)\le2\left(\epsilon\log|\my|+h_b(\epsilon)\right),\label{eq:le-c2-1}
\ee
where $[\sim q]:=[1:n]\backslash\{q\}$.

Also, for any random variable $Q\in[1:n]$ independent of $(X^n,Y^n)$ we have
\be
I(Y_Q;Q|X_Q)\le 2(\epsilon\log|\my|+h_b(\epsilon))\label{eq:le-c2-2}
.\ee
\end{lemma}

\subsection{Epsilon rate region}
\begin{lemma}
For all $\epsilon>0$, the simulation rate region is a subset of the set $\ms_{\epsilon}(r)$ which is as the set of all non-negative rate tuples $(R_0,R_{12},R_{21})$ for which there exists $p(f_1,\cdots,f_r,x_{[1:2]},y_{[1:2]})\in T_{\epsilon}(r)$ such that:

\begin{align}
R_{12}&\ge I(X_1;\F|X_2),\n
R_{21}&\ge I(X_2;\F|X_1),\n
R_0+R_{12}&\ge I(X_1;\F|X_2)+I(F_1;Y_{[1:2]}|X_{[1:2]})-3g(\epsilon),\n
R_0+R_{12}+R_{21}&\ge I(X_1;\F|X_2)+I(X_2;\F|X_1)+I(\F;Y_{[1:2]}|X_{[1:2]})-3g(\epsilon),
\end{align}
where $g(\epsilon):=2\left(\epsilon\log|\my_{[1:2]}|+h_b(\epsilon)\right)$ and $T_{\epsilon}(r)$ is the set of $p(f_1,\cdots,f_r,x_{[1:2]},y_{[1:2]})$ satisfying
\begin{align}
&\left\|p(x_{[1:2]},y_{[1:2]})-q(x_{[1:2]})q(y_{[1:2]}|x_{[1:2]})\right\|<\epsilon,\n
&\qquad\qquad F_i-F_{[1:i-1]}X_1-X_2, \ \mbox{if $i$ is odd,} \n
&\qquad\qquad  F_i-F_{[1:i-1]}X_2-X_1, \ \mbox{if $i$ is even,} \n
&\qquad\qquad Y_1-\F X_1-X_2Y_2,\n
&\qquad\qquad Y_2-\F X_2-X_1Y_1\n
\forall i:&\qquad
|\mf_i|\le|\mx_1||\mx_2||\my_1||\my_2|\prod_{j=1}^{i-1}|\mf_j|+1.
\end{align}
\end{lemma}
\begin{proof}
Without loss of generality, we can relax the cardinality bound from the definition of $T_{\epsilon}(r)$; an application of Fenchel-Carathe\`{o}dory theorem implies that the region $\ms_{\epsilon}(r)$ does not enlarge with this relaxation. See Appendix \ref{apx:cardinality} for the proof of cardinality bounds.

Take a random variable $Q$ uniform on $[1:n]$ and independent of all other random variables. Define $F_i=\omega C_i X_{1}^{Q+1:n}X_2^{1:Q-1}Q$ for $1\le i\le r$ and $X_i=X_{iQ},Y_i=Y_{iQ}$ for $i=1,2$. In the first step of the proof, we show the Markov chain conditions given in the definition of $T_{\epsilon}(r)$ are satisfied by this choice of auxiliary r.v.'s. These conditions are equivalent with the following
\begin{align}
C_i-&\omega C_{[1:i-1]}X_{1}^{q:n}X_2^{1:q-1}-X_{2,q} \ \mbox{if $i$ is odd,}\n
C_i-&\omega C_{[1:i-1]}X_{1}^{q+1:n}X_2^{1:q}-X_{1,q} \ \mbox{if $i$ is even,}\n
Y_{1,q}-&\omega C_{[1:r]}X_{1}^{q:n}X_2^{1:q-1}-X_{2,q}Y_{2,q},\n
Y_{2,q}-&\omega C_{[1:r]}X_{1}^{q+1:n}X_2^{1:q}-X_{1,q}Y_{1,q}.\label{eqn:apx-F}
\end{align}
\par The proof is provided in Appendix \ref{apx:c}. 

We know that
\[\tv{\tilde{p}(x_{[1:2]}^n,y_{[1:2]}^n)-q(x_{[1:2]}^n,y_{[1:2]}^n)}<\epsilon,\]
where $\tilde{p}(x_{[1:2]}^n,y_{[1:2]}^n)$ is the induced pmf of the code. This implies that for any value of $Q=q$,
\[\tv{\tilde{p}(x_{[1:2],q},y_{[1:2],q})-q(x_{[1:2]},y_{[1:2]})}<\epsilon,\]
therefore the total variation distance between the average of $\tilde{p}(x_{[1:2],q},y_{[1:2],q})$ over $Q=q$ (i.e. $\tilde{p}(x_{[1:2],Q},y_{[1:2],Q})$) and $q(x_{[1:2]},y_{[1:2]})$ is small, that is
\[\tv{\tilde{p}(x_{[1:2],Q},y_{[1:2],Q})-q(x_{[1:2]},y_{[1:2]})}<\epsilon.\]

Next we have
\begin{align}
nR_{12}&\ge \sum_{i:odd}H(C_i)\n&
\ge \sum_{i:odd}I(C_i;X_1^n|C_{[1:i-1]}X_2^n\omega)\n
&= \sum_{i=1}^rI(C_i;X_1^n|C_{[1:i-1]}X_2^n\omega)\label{eq:aA1}\\
&=I(C_{[1:r]};X_1^n|X_2^n\omega)\n&
=I(\omega C_{[1:r]};X_1^n|X_2^n)\label{eq:a2}\\
&=\sum_{q=1}^n I(\omega C_{[1:r]};X_{1,q}|X_{1}^{q+1:n}X_2^n)\n
&=\sum_{q=1}^n I(\omega C_{[1:r]}X_{1}^{q+1:n}X_{2,\sim q};X_{1,q}|X_{2,q})\\
&\ge \sum_{q=1}^n I(\omega C_{[1:r]}X_{1}^{q+1:n}X_{2}^{1:q-1};X_{1,q}|X_{2,q})\n
&=nI(\omega C_{[1:r]}X_{1}^{Q+1:n}X_{2}^{1:Q-1};X_{1,Q}|X_{2,Q},Q) \\
&=nI(\omega C_{[1:r]}X_{1}^{Q+1:n}X_{2}^{1:Q-1}Q;X_{1,Q}|X_{2,Q})\label{eq:a3}\\
&=nI(\F;X_{1}|X_{2})\label{eq:a4}
\end{align}
where \eqref{eq:aA1} follows from the Markov chain $C_i-C_{[1:i-1]}X_2^n\omega-X_1^n$ for even $i$, \eqref{eq:a2} is due to the independence of common randomness $\omega$ from $X_1^nX_2^n$ and the rest of the equations follow from the fact that $X_{1q},X_{2q}$ are i.i.d. repetitions. 

A similar statement can be proved for $R_{21}$:
\be
R_{21}\ge I(\F;X_{2}|X_{1}).\label{eq:b4}\ee
\vspace{.1cm}Next consider,
\begin{align}
n(R_{12}+R_0)&\ge H(\omega C_1)+\sum_{i:odd,~ i>1}H(C_i)\n &\ge H(\omega C_1|X_2^n)+\sum_{i:odd,~ i>1}H(C_i|C_{[1:i-1]}X_2^n\omega)\n
&\ge I(\omega C_1;Y_{[1:2]}^nX_1^n|X_2^n)+\sum_{i:odd,\ i>1}I(C_i;X_1^n|C_{[1:i-1]}X_2^n\omega)\n
&= I(\omega C_1;Y_{[1:2]}^n|X_{[1:2]}^n)+I(\omega C_1;X_1^n|X_2^n)+\sum_{i>1}I(C_i;X_1^n|C_{[1:i-1]}X_2^n\omega)\label{eq:b0.5}\\
&= I(\omega C_1;Y_{[1:2]}^n|X_{[1:2]}^n)+I(\omega C_1;X_1^n|X_2^n)+I(C_{[2:r]};X_1^n|X_2^nC_1\omega)\n
&{=} I(\omega C_1;Y_{[1:2]}^n|X_{[1:2]}^n)+I(C_{[1:r]}\omega;X_1^n|X_2^n)\n
&\ge I(\omega C_1;Y_{[1:2]}^n|X_{[1:2]}^n)+nI(\F;X_{1}|X_{2}),\label{eq:b1}
\end{align}
where \eqref{eq:b0.5} follows from the Markov chain $C_i-C_{[1:i-1]}X_2^n\omega-X_1^n$ for even $i$. Equation \eqref{eq:b1} follows equality of equations \eqref{eq:a2} and \eqref{eq:a4}. Now, we work out the first term of equation \eqref{eq:b1}.
\begin{align}
I(\omega C_1;Y_{[1:2]}^n|X_{[1:2]}^n)&=\sum_{q=1}^n I(\omega C_1;Y_{[1,2],q}|X_{[1:2]}^n,Y_{[1:2]}^{1:q-1})\n&=
\sum_{q=1}^nI(\omega C_1Y_{[1:2]}^{1:q-1};Y_{[1,2],q}|X_{[1:2]}^n)
-\sum_{q=1}^n I(Y_{[1:2]}^{1:q-1};Y_{[1:2],q}|X_{[1:2]}^n)\n
&\stackrel{(a)}\ge \sum_{q=1}^nI(\omega C_1;Y_{[1,2],q}|X_{[1:2]}^n)-ng(\epsilon)\n
& =\sum_{q=1}^nI(\omega C_1X_{[1:2],\sim q};Y_{[1:2],q}|X_{[1:2],q})-\sum_{q=1}^nI(X_{[1:2],\sim q};Y_{[1:2],q}|X_{[1:2],q})-ng(\epsilon)\n
&\stackrel{(b)}{\ge} \sum_{q=1}^nI(\omega C_1X_{1}^{q+1:n}X_2^{1:q-1};Y_{[1:2],q}|X_{[1:2],q})-2ng(\epsilon)\n
&=nI(\omega C_1X_{1}^{Q+1:n}X_2^{1:Q-1};Y_{[1:2],Q}|X_{[1:2],Q},Q)-2ng(\epsilon)\n
&=nI(Q\omega C_1X_{1}^{Q+1:n}X_2^{1:Q-1};Y_{[1:2],Q}|X_{[1:2],Q})-nI(Q;Y_{[1:2],Q}|X_{[1:2],Q})-2ng(\epsilon)\n
&\stackrel{(c)}{\ge} nI(F_1;Y_{[1:2]}|X_{[1:2]})-3ng(\epsilon),\label{eq:b50}
\end{align}
where  (a) is a result of Lemma \ref{le:c1}, and  (b) and  (c) follow from the Lemma \ref{le:c2}.

Equations \eqref{eq:b1} and \eqref{eq:b50} imply that 
\be\label{eq:c4}
R_{12}+R_0\ge I(F_1;Y_{[1:2]}|X_{[1:2]})+I(F_{[1:r]};X_{1}|X_{2})-3g(\epsilon).
\ee
\par Following the same lines as in the previous cases, we can show that
\begin{align}
n(R_0+R_{12}+R_{21})&\ge H(\omega C_1|X_{[1:2]}^n)+\sum_{i>1,\ i:odd}H(C_i|\omega C_{[1:i-1]}X_2^n)+\sum_{i:even}H(C_i|\omega C_{[1:i-1]}X_1^n)\n
&\ge I(\omega C_1;Y_{[1:2]}^nX_1^n|X_2^n)+\sum_{i:odd\atop i>1}I(C_i;Y_{[1:2]}^nX_1^n|\omega C_{[1:i-1]}X_2^n)+\sum_{i:even}I(C_i;Y_{[1:2]}^nX_2^n|\omega C_{[1:i-1]}X_1^n)\n
&= I(\omega C_1;X_1^n|X_2^n)+ I(\omega C_1;Y_{[1:2]}^n|X_{[1:2]}^n)\n&\qquad+\sum_{i:odd\atop i>1}\left[I(C_i;X_1^n|\omega C_{[1:i-1]}X_2^n)+I(C_i;Y_{[1:2]}^n|\omega C_{[1:i-1]}X_{[1:2]}^n)\right]\n&\qquad +
\sum_{i:even}\left[I(C_i;X_2^n|\omega C_{[1:i-1]}X_1^n)+I(C_i;Y_{[1:2]}^n|\omega C_{[1:i-1]}X_{[1:2]}^n)\right]\n
&= I(\omega C_1;X_1^n|X_2^n)+\sum_{i:odd\atop i>1}I(C_i;X_1^n|\omega C_{[1:i-1]}X_2^n)+\sum_{i:even}I(C_i;X_2^n|\omega C_{[1:i-1]}X_1^n)\n&\qquad
+ I(\omega C_1;Y_{[1:2]}^n|X_{[1:2]}^n)+\sum_{i:odd\atop i>1}I(C_i;Y_{[1:2]}^n|\omega C_{[1:i-1]}X_{[1:2]}^n)+
\sum_{i:even}I(C_i;Y_{[1:2]}^n|\omega C_{[1:i-1]}X_{[1:2]}^n)\n
&= I(\omega C_1;X_1^n|X_2^n)+\sum_{i:odd\atop i>1}I(C_i;X_1^n|\omega C_{[1:i-1]}X_2^n)+\sum_{i:even}I(C_i;X_2^n|\omega C_{[1:i-1]}X_1^n)\n&\qquad
+ I(\omega C_1;Y_{[1:2]}^n|X_{[1:2]}^n)+\sum_{i>1}I(C_i;Y_{[1:2]}^n|\omega C_{[1:i-1]}X_{[1:2]}^n)\n
&= I(\omega C_1;X_1^n|X_2^n)+\sum_{i>1}I(C_i;X_1^n|\omega C_{[1:i-1]}X_2^n)+\sum_{i}I(C_i;X_2^n|\omega C_{[1:i-1]}X_1^n)\n&\qquad
+ I(\omega C_{[1:r]};Y_{[1:2]}^n|X_{[1:2]}^n)
\n&=I(\omega C_{[1:r]};X_1^n|X_2^n)+I(\omega C_{[1:r]};X_2^n|X_1^n)+I(\omega C_{[1:r]};Y_{[1:2]}^n|X_{[1:2]}^n)\n
&\ge n(I(\F;X_{1,Q}|X_{2,Q})+I(\F;X_{2,Q}|X_{1,Q})+I(\F;Y_{[1:2],Q}|X_{[1:2],Q})-3g(\epsilon))\label{eq:b51}
\end{align}
where the first term of \eqref{eq:b51} follows from equality of equations \eqref{eq:a2} and \eqref{eq:a4}, second term follows similarly and the last term follows from an argument similar to the one given in deriving equation \eqref{eq:b50}.
\end{proof}
\subsection{Continuity of $\ms_{\epsilon}(r)$ at $\epsilon=0$}
\begin{lemma}
\[
\ms(r)=\bigcap_{\epsilon>0}\ms_{\epsilon}(r).
\]
\end{lemma}
\begin{proof}
It is clear that $\ms(r)\subseteq\bigcap_{\epsilon>0}\ms_{\epsilon}(r)$. We now prove the reverse direction, i.e., $\bigcap_{\epsilon>0}\ms_{\epsilon}(r)\subseteq\ms(r)$. To show this, we take a vanishing sequence $\{\epsilon_k\}_{k\ge1}$. Take a point $\mathbf{R}^*=(R^*_0,R^*_{12},R^*_{21})$ in $\cap_{k\ge1}\ms_{\epsilon_k}(r)$. Corresponding to this point is a sequence of pmfs $p_k(f_{1:r},x_{[1:2]},y_{[1:2]})\in T_{\epsilon_k}(r)$. Since these pmfs  belong to the probability simplex $\Delta^{|\mf_{[1:r]}||\mx_{[1:2]}||\my_{[1:2]}|}$ and the probability simplex is compact (due to the cardinality bounds on $\mf_i,~1\le i\le r$, there exists a sequence $\{{i_k}\}_{k\ge1}$ such that the sequence $p_{i_k}(f_{1:r},x_{[1:2]},y_{[1:2]})$ converges to some $p^*(f_{1:r},x_{[1:2]},y_{[1:2]})$ in the probability simplex. $p^*(f_{1:r},x_{[1:2]},y_{[1:2]})$ must belong to $T(r)$, because total variation distance and mutual information function are continuous in the probability simplex. In particular, we have 
\begin{small}\begin{align*}
 \tv{p^*(x_{[1:2]},y_{[1:2]})-q(x_{[1:2]},y_{[1:2]})}=\lim_{k\rightarrow\infty}\tv{p_{i_k}(x_{[1:2]},y_{[1:2]})-q(x_{[1:2]},y_{[1:2]})}=0&\Rightarrow p^*(x_{[1:2]},y_{[1:2]})=q(x_{[1:2]},y_{[1:2]}),\\
 \mbox{$i$ is odd:}~~I_{p^*}(F_i;X_2|F_{[1:i-1}X_1)=\lim_{k\rightarrow\infty}I_{p_{i_k}}(F_i;X_2|F_{[1:i-1}X_1)=0&\Rightarrow F_i-F_{[1:i-1]}X_1-X_2,\\
 \mbox{$i$ is even:}~~I_{p^*}(F_i;X_1|F_{[1:i-1}X_2)=\lim_{k\rightarrow\infty}I_{p_{i_k}}(F_i;X_1|F_{[1:i-1}X_2)=0&\Rightarrow F_i-F_{[1:i-1]}X_2-X_1,\\
 I_{p^*}(Y_1;X_2Y_2|F_{[1:r]}X_1)=\lim_{k\rightarrow\infty}I_{p_{i_k}}(Y_1;X_2Y_2|F_{[1:r]}X_1)=0&\Rightarrow Y_1-F_{[1:r]}X_1-X_2Y_2,\\
 I_{p^*}(Y_2;X_1Y_1|F_{[1:r]}X_2)=\lim_{k\rightarrow\infty}I_{p_{i_k}}(Y_2;X_1Y_1|F_{[1:r]}X_2)=0&\Rightarrow Y_2-F_{[1:r]}X_2-X_1Y_1.
\end{align*}
\end{small}
Further one can show that $\mathbf{R}^*$ is a point of $\ms(r)$ corresponded to the pmf $p^*(f_{[1:r]},x_{[1:2]},y_{[1:2]})$. This is because $\lim_{\epsilon\rightarrow 0}g(\epsilon)=0$ and the mutual information terms defining the set $\ms_{\epsilon_{i_k}}(r)$ tends to the ones corresponded to $p^*(f_{[1:r]},x_{[1:2]},y_{[1:2]})$. This concludes the proof. 
\end{proof}

\appendix
\section{Inductive proof of the approximation \eqref{eqn:apx-A}}\label{apx:a}
In this appendix we find the constraints that imply 
\begin{align*}
P(x_{[1:2]}^n,f_{[1:r]}^n,b_{[1:r]}, k_{[1:r]}, \omega_{[1:r]}, \hat{f}_{[1:r],\ta}^n, \hat{f}_{[1:r],\tb}^n)\apx{\epsilon_n}
\widehat{P}(x_{[1:2]}^n,f_{[1:r]}^n,b_{[1:r]}, k_{[1:r]}, \omega_{[1:r]}, \hat{f}_{[1:r],\ta}^n, \hat{f}_{[1:r],\tb}^n).
\end{align*}
Let $Z_0=X_{[1:2]}^n$ and $Z_j=(F_j^n, B_j, K_j, \omega_j, \hat{F}_{j,\ta}^n, \hat{F}_{j,\tb}^n)$ for $j\in[1:r]$. For any $j\in[0:r]$ we inductively find constraints that imply
\begin{align}\label{eq-apx:1}
P(Z_{[1:j]})\apx{\epsilon_n^{(j)}}
\widehat{P}(Z_{[1:j]}),
\end{align}
for some $\epsilon_n^{(j)}$ converging to zero as $n\rightarrow\infty$. 

Let us define a new random pmf $\widetilde{P}$ by changing one of the terms in the expansion of the pmf $P$ of the protocol A given in \eqref{eq:3}. We replace the Slepian-Wolf terms with one that corresponds to an ideal zero probability of error.
\begin{align}
\widetilde{P}(x_{[1:2]}^n,&f_{[1:r]}^n,b_{[1:r]}, k_{[1:r]}, \omega, y_{[1:2]}^n, \hat{f}_{[1:r],\ta}^n, \hat{f}_{[1:r],\tb}^n)\n&=
p(x_{[1:2]}^n)
\left[\prod_{i=1}^rP(b_i,\omega_i|f_{[1:i-1]}^n, x_{(i)_2}^n)P(f_i^n,k_i|b_i,\omega_i,f_{[1:i-1]}^n, x_{(i)_2}^n)\right.\n&\qquad\qquad\qquad\left. \mathbf{1}\{\hat{f}^n_{i,\t{i+1}}=f^n_{i}\}\mathbf{1}\{\hat{f}^n_{i,\t{i}}=f^n_{i}\}\right] {p}(y_1^n|f_{[1:r]}^nx_1^n)p(y_2^n|f_{[1:r]}^nx_2^n).\label{eq:apxa0}
\end{align}  
In order to show that the pmfs $P$ and $\widehat{P}$ in \eqref{eq-apx:1} are close, we show that both are close to $\widetilde{P}$. Therefore they have to be also close to each other because of the triangle inequality. In other words we will inductively find constraints that imply 
\begin{align}\label{eq-apx:2}
P(Z_{[1:k]})\apx{\epsilon_n^{(k)}}\widetilde{P}(Z_{[1:k]}),\qquad \mbox{for $k\in[0:r]$},\n
\widehat{P}(Z_{[1:k]})\apx{\epsilon_n^{(k)}}\widetilde{P}(Z_{[1:k]}),\qquad \mbox{for $k\in[0:r]$},
\end{align}
for some $\epsilon_n^{(k)}$ converging to zero as $n\rightarrow\infty$. For $j=0$ this is trivial since it reduces to $P(x_{[1:2]}^n)=p(x_{[1:2]}^n)=\widehat{P}(x_{[1:2]}^n)=\widetilde{P}(x_{[1:2]}^n)$. Suppose that \eqref{eq-apx:2} holds for $k=j-1$. To show it for $k=j$ we proceed as follows. First observe that it suffices to prove the existence of a sequence $\delta_n\rightarrow 0$ such that
\begin{subequations}\label{eq:apxa2.5}\begin{align}
\widetilde{P}(Z_{[0:j]})=\widetilde{P}(Z_{[0:j-1]})\widetilde{P}(Z_j|Z_{[0:j-1]})&\apx{\delta_n}\widetilde{P}(Z_{[1:j-1]}){P}(Z_j|Z_{[0:j-1]}),\label{eq:apxa4}\\
\widetilde{P}(Z_{[0:j]})=\widetilde{P}(Z_{[0:j-1]})\widetilde{P}(Z_j|Z_{[0:j-1]})&\apx{\delta_n}\widetilde{P}(Z_{[0:j-1]})\widehat{P}(Z_j|Z_{[0:j-1]}),\label{eq:apxa3}
\end{align}
because the third part of Lemma \ref{le:total} then yields that
\begin{align*}
P(Z_{[0:j]})&\apx{\epsilon_n^{(j)}}
\widetilde{P}(Z_{[0:j]}),\\
\widehat{P}(Z_{[0:j]})&\apx{\epsilon_n^{(j)}}\widetilde{P}(Z_{[0:j]}),
\end{align*}
where $\epsilon_n^{(j)}=\epsilon_n^{(j-1)}+\delta_n$. Next, note that the triangle inequality implies that instead of showing  \eqref{eq:apxa3} one can show \eqref{eq:apxa4} and \eqref{eq:apxa13} given below\begin{align}
\widetilde{P}(Z_{[1:j-1]}){P}(Z_j|Z_{[0:j-1]})\apx{\delta_n}\widetilde{P}(Z_{[0:j-1]})\widehat{P}(Z_j|Z_{[0:j-1]}).\label{eq:apxa13}
\end{align}
\end{subequations}
Therefore it suffices to show \eqref{eq:apxa4} and \eqref{eq:apxa13}. 

We begin by finding the expressions for the terms appearing in \eqref{eq:apxa2.5}.
The marginal pmf $\widetilde{P}(Z_{[0:j-1]})$ (computed from equation \eqref{eq:apxa0}) is as follows: \begin{align}
\widetilde{P}(Z_{[0:j-1]})
                                    &=p(x_{[1:2]}^n)
\left[\prod_{i=1}^{j-1}P(b_i,\omega_i|f_{[1:i-1]}^n, x_{(i)_2}^n)P(f_i^n,k_i|b_i,\omega_i,f_{[1:i-1]}^n, x_{(i)_2}^n)\right.\n&\qquad\qquad\qquad\qquad\qquad\qquad\qquad\qquad\left.\times \mathbf{1}\{\hat{f}^n_{i,\t{i+1}}=f^n_{i}\}\mathbf{1}\{\hat{f}^n_{i,\t{i}}=f^n_{i}\}\right]\label{eq:apxa-2}\\
&=p(x_{[1:2]}^n)
\left[\prod_{i=1}^{j-1}P(f_i^n|f_{[1:i-1]}^n, x_{(i)_2}^n)P(b_i,\omega_i,k_i|f_{[1:i]}^n, x_{(i)_2}^n)\right.\n&\qquad\qquad\qquad\qquad\qquad\qquad\qquad\qquad\left.\times \mathbf{1}\{\hat{f}^n_{i,\t{i+1}}=f^n_{i}\}\mathbf{1}\{\hat{f}^n_{i,\t{i}}=f^n_{i}\}\right]\n
&=p(x_{[1:2]}^n)
\left[\prod_{i=1}^{j-1}p(f_i^n|f_{[1:i-1]}^n, x_{(i)_2}^n)P(b_i,\omega_i,k_i|f_{[1:i]}^n)\right.\n&\qquad\qquad\qquad\qquad\qquad\qquad\qquad\qquad\left.\times \mathbf{1}\{\hat{f}^n_{i,\t{i+1}}=f^n_{i}\}\mathbf{1}\{\hat{f}^n_{i,\t{i}}=f^n_{i}\}\right]\label{eq:apxa-1.5}\\
&=p(x_{[1:2]}^n,f_{[1:j-1]}^n)
\left[\prod_{i=1}^{j-1}P(b_i,\omega_i,k_i|f_{[1:i]}^n)\right.\n&\qquad\qquad\qquad\qquad\qquad\qquad\qquad\qquad\left.\times \mathbf{1}\{\hat{f}^n_{i,\t{i+1}}=f^n_{i}\}\mathbf{1}\{\hat{f}^n_{i,\t{i}}=f^n_{i}\}\right]
\label{eq:apxa-2}
,\end{align}
where equation \eqref{eq:apxa-1.5} follows from the fact that $(b_i,\omega_i,k_i)$ are (random) bin indices of $f_{[1:i]}^n$ in Protocol A, and that $(F_i^n,F_{[1:i-1]}^n, X_{(i)_2}^n)$ have an i.i.d. pmf in the same protocol. Equation \eqref{eq:apxa-2} follows from the Markov conditions on $X_{[1:2]}^n, F_{[1:j-1]}^n$ in $T(r)$.

There are three conditional pmfs in \eqref{eq:apxa2.5} that can be computed from equations \eqref{eq:apxa0}, \eqref{eq:a3}, \eqref{eq:a4} respectively as follows:
\begin{align}
\widetilde{P}(Z_j|Z_{[0:j-1]})&=P(b_j,\omega_j|f_{[1:j-1]}^n, x_{(j)_2}^n)P(f_j^n,k_j|b_j,\omega_j,f_{[1:j-1]}^n, x_{(j)_2}^n)\n&\quad\qquad\qquad\qquad\qquad\times \mathbf{1}\{\hat{f}^n_{j,\t{j+1}}=f^n_{j}\}\mathbf{1}\{\hat{f}^n_{j,\t{j}}=f^n_{j}\},\label{eq:apx5}
\\
{P}(Z_j|Z_{[0:j-1]})&=P(b_j,\omega_j|f_{[1:j-1]}^n, x_{(j)_2}^n)P(f_j^n,k_j|b_j,\omega_j,f_{[1:j-1]}^n, x_{(j)_2}^n)\n&\quad\qquad\qquad\qquad\qquad\times P^{S-W}(\hat{f}_{j,\t{j+1}}^n|b_j,k_j,\omega_j,f_{[1:j-1]}^n,x^n_{(j+1)_2})\mathbf{1}\{\hat{f}^n_{j,\t{j}}=f^n_{j}\}\label{eq:apx7},
\\
\widehat{P}(Z_j|Z_{[0:j-1]})&=p^U(\omega_j)p^U(b_{j}) P(f_j^n,k_j|b_j,\omega_j,\hat{f}_{[1:j-1],\t{j}}^n, x_{(j)_2}^n)\n&\quad\qquad\qquad\qquad\qquad \times P^{S-W}(\hat{f}_{j,\t{j+1}}^n|b_j,k_j,\omega_j,\hat{f}_{[1:j-1],\t{j+1}}^n,x^n_{(j+1)_2})\mathbf{1}\{\hat{f}^n_{j,\t{j}}=f^n_{j}\}.\label{eq:apx6}
\end{align}

\emph{Finding sufficient conditions for equation \eqref{eq:apxa4} to hold:}

We begin by showing equation \eqref{eq:apxa4}. Note that the only difference in the two pmf expressions is that the Slepian-Wolf term in \eqref{eq:apx7} is replaced with an indicator function in \eqref{eq:apx5}. To use Slepian-Wolf theorem we need to show that we are dealing with an i.i.d. scenario where random bin indices are transmitted from one party to another party. Let us rewrite equations \eqref{eq:apx7} and \eqref{eq:apx5} as follows:
\begin{align}
\widetilde{P}(Z_j|Z_{[0:j-1]})&=P(f_j^n|f_{[1:j-1]}^n, x_{(j)_2}^n)
P(b_j,\omega_j,k_j|f_{[1:j]}^n, x_{(j)_2}^n)\n&\quad\qquad\qquad\qquad\qquad\times \mathbf{1}\{\hat{f}^n_{j,\t{j+1}}=f^n_{j}\}\mathbf{1}\{\hat{f}^n_{j,\t{j}}=f^n_{j}\}\n&
=p(f_j^n|f_{[1:j-1]}^n, x_{(j)_2}^n)
P(b_j,\omega_j,k_j|f_{[1:j]}^n)\n&\quad\qquad\qquad\qquad\qquad\times \mathbf{1}\{\hat{f}^n_{j,\t{j+1}}=f^n_{j}\}\mathbf{1}\{\hat{f}^n_{j,\t{j}}=f^n_{j}\},\label{eq:apx8}
\\
{P}(Z_j|Z_{[0:j-1]})&=P(f_j^n|f_{[1:j-1]}^n, x_{(j)_2}^n)
P(b_j,\omega_j,k_j|f_{[1:j]}^n, x_{(j)_2}^n)\n&\quad\qquad\qquad\qquad\qquad\times P^{S-W}(\hat{f}_{j,\t{j+1}}^n|b_j,k_j,\omega_j,f_{[1:j-1]}^n,x^n_{(j+1)_2})\mathbf{1}\{\hat{f}^n_{j,\t{j}}=f^n_{j}\}
\n&=p(f_j^n|f_{[1:j-1]}^n, x_{(j)_2}^n)
P(b_j,\omega_j,k_j|f_{[1:j]}^n)\n&\quad\qquad\qquad\qquad\qquad\times P^{S-W}(\hat{f}_{j,\t{j+1}}^n|b_j,k_j,\omega_j,f_{[1:j-1]}^n,x^n_{(j+1)_2})\mathbf{1}\{\hat{f}^n_{j,\t{j}}=f^n_{j}\}
\label{eq:apx9}.
\end{align}
We now compute $\widetilde{P}(Z_{[0:j-1]})\widetilde{P}(Z_j|Z_{[0:j-1]})$ and $\widetilde{P}(Z_{[0:j-1]}){P}(Z_j|Z_{[0:j-1]})$ using equation \eqref{eq:apxa-2} as follows:
 \begin{align}
\widetilde{P}(Z_{[0:j-1]})\widetilde{P}(Z_j|Z_{[0:j-1]})&=p(x_{[1:2]}^n,f_{[1:j]}^n)
\left[\prod_{i=1}^{j-1}P(b_i,\omega_i,k_i|f_{[1:i]}^n) \mathbf{1}\{\hat{f}^n_{i,\t{i+1}}=f^n_{i}\}\mathbf{1}\{\hat{f}^n_{i,\t{i}}=f^n_{i}\}\right]\\
&\quad\times P(b_j,\omega_j,k_j|f_{[1:j]}^n)\mathbf{1}\{\hat{f}^n_{j,\t{j}}=f^n_{j}\}
\n&\quad\times\mathbf{1}\{\hat{f}^n_{i,\t{i+1}}=f^n_{i}\}\label{eq:apxa10},\\
\widetilde{P}(Z_{[0:j-1]}){P}(Z_j|Z_{[0:j-1]})&=p(x_{[1:2]}^n,f_{[1:j]}^n)
\left[\prod_{i=1}^{j-1}P(b_i,\omega_i,k_i|f_{[1:i]}^n) \mathbf{1}\{\hat{f}^n_{i,\t{i+1}}=f^n_{i}\}\mathbf{1}\{\hat{f}^n_{i,\t{i}}=f^n_{i}\}\right]\n
&\quad\times 
P(b_j,\omega_j,k_j|f_{[1:j]}^n)\mathbf{1}\{\hat{f}^n_{j,\t{j}}=f^n_{j}\}\n
&\quad\times
P^{S-W}(\hat{f}_{j,\t{j+1}}^n|b_j,k_j,\omega_j,f_{[1:j-1]}^n,x^n_{(j+1)_2})
\label{eq:apxa11}.
\end{align}
Using the first part of Lemma \ref{le:total} it suffices to show that
 \begin{align*}
&p(x_{[1:2]}^n,f_{[1:j]}^n)P(b_j,k_j,\omega_j|f_{[1:j]}^n)
\mathbf{1}\{\hat{f}^n_{i,\t{i+1}}=f^n_{i}\}\apx{\delta_n}\\&p(x_{[1:2]}^n,f_{[1:j]}^n)P(b_j,k_j,\omega_j|f_{[1:j]}^n)
P^{S-W}(\hat{f}_{j,\t{j+1}}^n|b_j,k_j,\omega_j,f_{[1:j-1]}^n,x^n_{(j+1)_2}).
\end{align*}
The above pmf corresponds to an Slepian-Wolf problem where the first party has i.i.d. repetitions $(f_{[1:j-1]}^n, x^n_{(j)_2})$ and the second party has i.i.d. repetitions $(f_{[1:j-1]}^n, x^n_{(j+1)_2})$. The first party creates i.i.d. repetitions $f_j^n$ and communicates random bin indices $b_j,k_j,\omega_j$ of $f_{[1:j]}^n$ to the second party. Using Lemma \ref{le:S-W} the above total variation is small as long as the following constraints hold: 
\begin{itemize}\item For $j=1$, $\omega_j$ is non-empty and the S-W decoding is reliable if,
\be
R_1+R_0+\tR_1\ge H(F_1|X_2).
\ee
\item
For $j\geq 2$ the S-W decoding is reliable if,
\be
R_j+\tR_j\ge H(F_{[1:j]}|X_{(j+1)_2}F_{[1:j-1]})=H(F_j|X_{(j+1)_2}F_{[1:j-1]}).
\ee
\end{itemize}

\emph{Finding sufficient conditions for equation \eqref{eq:apxa13} to hold:}

The pmf $\widetilde{P}(Z_{[0:j-1]})P(Z_j|Z_{[0:j-1]})$ was computed in equation \eqref{eq:apxa11}. We now compute  $\widetilde{P}(Z_{[0:j-1]})\widehat{P}(Z_j|Z_{[0:j-1]})$ using equations \eqref{eq:apxa-2} and \eqref{eq:apx6} as follows:
 {\allowdisplaybreaks\begin{align}
\widetilde{P}(Z_{[0:j-1]}){P}(Z_j|Z_{[0:j-1]})&=p(x_{[1:2]}^n,f_{[1:j-1]}^n)
\left[\prod_{i=1}^{j-1}P(b_i,\omega_i,k_i|f_{[1:i]}^n) \mathbf{1}\{\hat{f}^n_{i,\t{i+1}}=f^n_{i}\}\mathbf{1}\{\hat{f}^n_{i,\t{i}}=f^n_{i}\}\right]\n
&\quad\times 
p^U(\omega_j)p^U(b_{j}) P(f_j^n,k_j|b_j,\omega_j,\hat{f}_{[1:j-1],\t{j}}^n, x_{(j)_2}^n)\mathbf{1}\{\hat{f}^n_{j,\t{j}}=f^n_{j}\}\n
&\quad\times
P^{S-W}(\hat{f}_{j,\t{j+1}}^n|b_j,k_j,\omega_j,\hat{f}_{[1:j-1],\t{j+1}}^n,x^n_{(j+1)_2})\n
&=p(x_{[1:2]}^n,f_{[1:j-1]}^n)
\left[\prod_{i=1}^{j-1}P(b_i,\omega_i,k_i|f_{[1:i]}^n) \mathbf{1}\{\hat{f}^n_{i,\t{i+1}}=f^n_{i}\}\mathbf{1}\{\hat{f}^n_{i,\t{i}}=f^n_{i}\}\right]\n
&\quad\times 
p^U(\omega_j)p^U(b_{j}) P(f_j^n,k_j|b_j,\omega_j,f_{[1:j-1]}^n, x_{(j)_2}^n)\mathbf{1}\{\hat{f}^n_{j,\t{j}}=f^n_{j}\}\n
&\quad\times
P^{S-W}(\hat{f}_{j,\t{j+1}}^n|b_j,k_j,\omega_j,f_{[1:j-1]}^n,x^n_{(j+1)_2})\label{eq:apxa110}\\
&=p(x_{[1:2]}^n,f_{[1:j-1]}^n)
\left[\prod_{i=1}^{j-1}P(b_i,\omega_i,k_i|f_{[1:i]}^n) \mathbf{1}\{\hat{f}^n_{i,\t{i+1}}=f^n_{i}\}\mathbf{1}\{\hat{f}^n_{i,\t{i}}=f^n_{i}\}\right]\n
&\quad\times 
p^U(\omega_j)p^U(b_{j}) P(f_j^n|b_j,\omega_j,f_{[1:j-1]}^n, x_{(j)_2}^n)P(k_j|f_{[1:j]}^n)\mathbf{1}\{\hat{f}^n_{j,\t{j}}=f^n_{j}\}\n
&\quad\times
P^{S-W}(\hat{f}_{j,\t{j+1}}^n|b_j,k_j,\omega_j,f_{[1:j-1]}^n,x^n_{(j+1)_2}),\label{eq:apxa111}
\end{align}}
where \eqref{eq:apxa110} holds since $\hat{F}_{[1:j-1]}^n=F_{[1:j-1]}^n$ holds because of the \emph{indicator functions} in $\widetilde{P}(Z_{[0:j-1]})$; equation \eqref{eq:apxa111} holds since $k_j$ is a (random) bin index of $f_{[1:j]}^n$.

Let us compare \eqref{eq:apxa111} and \eqref{eq:apxa11}. We see that most of the terms are the same. Using the first part of Lemma \ref{le:total} it suffices to show that
 \begin{align}
&p(x_{[1:2]}^n,f_{[1:j-1]}^n)p(f_j^n|x_{[1:2]}^n,f_{[1:j-1]}^n)P(b_j,\omega_j|f_{[1:j]}^n)
\apx{\delta_n}\nonumber\\&p(x_{[1:2]}^n,f_{[1:j-1]}^n)p^U(\omega_j)p^U(b_{j}) P(f_j^n|b_j,\omega_j,f_{[1:j-1]}^n, x_{(j)_2}^n).\label{eqn:apxa112}
\end{align}
Note that 
\begin{align}
&p(x_{[1:2]}^n,f_{[1:j-1]}^n)p(f_j^n|x_{[1:2]}^n,f_{[1:j-1]}^n)P(b_j,\omega_j|f_{[1:j]}^n)
=\nonumber\\&p(x_{[1:2]}^n,f_{[1:j-1]}^n)P(b_j,\omega_j|f_{[1:j-1]}^n,x_{[1:2]}^n) P(f_j^n|b_j,\omega_j,f_{[1:j-1]}^n, x_{(j)_2}^n).\label{eqn:apxa113}
\end{align}
We note that $P(f_j^n|b_j,\omega_j,f_{[1:j-1]}^n, x_{(j)_2}^n)$ of the above equation is the one of Protocol A and used in Protocol B. Now, to show that \eqref{eqn:apxa112} holds it suffices to show the following equation because the first part of Lemma \ref{le:total}:
 \begin{align}
p(x_{[1:2]}^n,f_{[1:j-1]}^n)P(b_j,\omega_j|f_{[1:j-1]}^n,x_{[1:2]}^n)
\apx{\delta_n}p(x_{[1:2]}^n,f_{[1:j-1]}^n)p^U(\omega_j)p^U(b_{j}).\label{eqn:apxa114}
\end{align}
In other words we need to impose constraints that imply $(B_j,\omega_j)$ are mutually nearly independent of $(F_{[1:j-1]}^n,X_{[1:2]}^n)$. Substituting $T=1$, $X_1=F_{[1:i]}$ and $Y=X_{[1:2]}F_{[1:i-1]}$ in Theorem \ref{thm:re} yields that equation \eqref{eqn:apxa114} holds if

\begin{itemize}
\item For $j=1$ we have the constraint
\be
R_0+\tR_1< H(F_1|X_1X_2)=H(F_1|X_1),\label{eq:apx20}
\ee
\item
for $j\geq 2$ we have the constraints
\be
\tR_i<H(F_j|X_{[1:2]}F_{[1:j-1]})=H(F_j|X_{(j)_2}F_{[1:j-1]}),\label{eq:apx30}
\ee
where in \eqref{eq:apx20} and \eqref{eq:apx30} we use the Markov chain $F_j-X_{(j)_2}F_{[1:j-1]}-X_{(j+1)_2}$ for any $j$.

This completes the induction proof.
\end{itemize}

\section{Proof of the approximation \eqref{eqn:apx-B}}\label{apx:b}
In this appendix we show that the approximation 
\begin{align}\label{aq:apxb0}
\widehat{P}(x_{[1:2]}^n,f_{[1:r]}^n,b_{[1:r]}, k_{[1:r]}, \omega_{[1:r]}, \hat{f}_{[1:r],\ta}^n, \hat{f}_{[1:r],\tb}^n)\apx{\epsilon_n}
{P}(x_{[1:2]}^n,f_{[1:r]}^n,b_{[1:r]}, k_{[1:r]}, \omega_{[1:r]}, \hat{f}_{[1:r],\ta}^n, \hat{f}_{[1:r],\tb}^n),
\end{align}
implies 
\begin{align}\label{aq:apxb2}
\widehat{P}(x_{[1:2]}^n,f_{[1:r]}^n,b_{[1:r]}, k_{[1:r]}, \omega,y_{[1:2]}^n,\hat{f}_{[1:r],\ta}^n, \hat{f}_{[1:r],\tb}^n)\apx{\tilde{\epsilon}_n}
{P}(x_{[1:2]}^n,f_{[1:r]}^n,b_{[1:r]}, k_{[1:r]}, \omega,y_{[1:2]}^n,\ \hat{f}_{[1:r],\ta}^n, \hat{f}_{[1:r],\tb}^n),
\end{align}
for some sequence $\tilde{\epsilon}_n\rightarrow 0$.

We prove it indirectly through the random pmf $\widetilde{P}$ introduced in \eqref{eq:apxa0}. It has been shown in Appendix \ref{apx:a} that in addition to the approximation \eqref{aq:apxb0} the following approximation holds
\begin{align}\label{aq:apxb1}
P(x_{[1:2]}^n,f_{[1:r]}^n,b_{[1:r]}, k_{[1:r]}, \omega_{[1:r]}, \hat{f}_{[1:r],\ta}^n, \hat{f}_{[1:r],\tb}^n)\apx{\epsilon_n}
\widetilde{P}(x_{[1:2]}^n,f_{[1:r]}^n,b_{[1:r]}, k_{[1:r]}, \omega_{[1:r]}, \hat{f}_{[1:r],\ta}^n, \hat{f}_{[1:r],\tb}^n).
\end{align}
Note that the triangle inequality implies that instead of  showing the approximation \eqref{aq:apxb2} one can show the following approximations
\begin{align}
P(x_{[1:2]}^n,f_{[1:r]}^n,b_{[1:r]}, k_{[1:r]}, \omega,y_{[1:2]}^n,\hat{f}_{[1:r],\ta}^n, \hat{f}_{[1:r],\tb}^n)&\apx{\tilde{\epsilon}_n}
\widetilde{P}(x_{[1:2]}^n,f_{[1:r]}^n,b_{[1:r]}, k_{[1:r]}, \omega,y_{[1:2]}^n, \hat{f}_{[1:r],\ta}^n, \hat{f}_{[1:r],\tb}^n),\n
\widehat{P}(x_{[1:2]}^n,f_{[1:r]}^n,b_{[1:r]}, k_{[1:r]}, \omega,y_{[1:2]}^n,\hat{f}_{[1:r],\ta}^n, \hat{f}_{[1:r],\tb}^n)&\apx{\tilde{\epsilon}_n}
\widetilde{P}(x_{[1:2]}^n,f_{[1:r]}^n,b_{[1:r]}, k_{[1:r]}, \omega,y_{[1:2]}^n, \hat{f}_{[1:r],\ta}^n, \hat{f}_{[1:r],\tb}^n).
\end{align}
Using the third part of Lemma \ref{le:total}, it suffices to prove the following approximations
\begin{align}
&\widetilde{P}(x_{[1:2]}^n,f_{[1:r]}^n,b_{[1:r]}, k_{[1:r]}, \omega,\hat{f}_{[1:r],\ta}^n, \hat{f}_{[1:r],\tb}^n)
P(y_{[1:2]}^n|x_{[1:2]}^n,f_{[1:r]}^n,b_{[1:r]}, k_{[1:r]}, \omega,\hat{f}_{[1:r],\ta}^n, \hat{f}_{[1:r],\tb}^n)\n&\qquad\quad\apx{\tilde{\epsilon}_n}
\widetilde{P}(x_{[1:2]}^n,f_{[1:r]}^n,b_{[1:r]}, k_{[1:r]}, \omega,y_{[1:2]}^n,\ \hat{f}_{[1:r],\ta}^n, \hat{f}_{[1:r],\tb}^n),\n
&\widetilde{P}(x_{[1:2]}^n,f_{[1:r]}^n,b_{[1:r]}, k_{[1:r]}, \omega,\hat{f}_{[1:r],\ta}^n, \hat{f}_{[1:r],\tb}^n)
\widehat{P}(y_{[1:2]}^n|x_{[1:2]}^n,f_{[1:r]}^n,b_{[1:r]}, k_{[1:r]}, \omega,\hat{f}_{[1:r],\ta}^n, \hat{f}_{[1:r],\tb}^n)\n&\qquad\quad\apx{\tilde{\epsilon}_n}
\widetilde{P}(x_{[1:2]}^n,f_{[1:r]}^n,b_{[1:r]}, k_{[1:r]}, \omega,y_{[1:2]}^n,\ \hat{f}_{[1:r],\ta}^n, \hat{f}_{[1:r],\tb}^n).\label{aq:apxb4}
\end{align}
First observe that 
\begin{align*}
P(y_{[1:2]}^n|x_{[1:2]}^n,f_{[1:r]}^n,b_{[1:r]}, k_{[1:r]}, \omega,\hat{f}_{[1:r],\ta}^n, \hat{f}_{[1:r],\tb}^n)&=\widetilde{P}(y_{[1:2]}^n|x_{[1:2]}^n,f_{[1:r]}^n,b_{[1:r]}, k_{[1:r]}, \omega,\hat{f}_{[1:r],\ta}^n, \hat{f}_{[1:r],\tb}^n)\n&={p}(y_1^n|f_{[1:r]}^nx_1^n)p(y_2^n|f_{[1:r]}^nx_2^n).
\end{align*}
This equation gives the first approximation of \eqref{aq:apxb4} with \emph{equality}. 

Next using equation \eqref{eq:4} we get
\begin{align}
\widehat{P}(y_{[1:2]}^n|x_{[1:2]}^n,f_{[1:r]}^n,b_{[1:r]}, k_{[1:r]}, \omega,\hat{f}_{[1:r],\ta}^n, \hat{f}_{[1:r],\tb}^n)={p}(y_1^n|\hat{f}_{[1:r],\ta}^nx_1^n)p(y_2^n|\hat{f}_{[1:r],\tb}^nx_2^n).
\end{align}
Substituting this in the second equation of \eqref{aq:apxb4} gives the second approximation of \eqref{aq:apxb4} with \emph{equality} as follows
\begin{align}
\widetilde{P}(&x_{[1:2]}^n,f_{[1:r]}^n,b_{[1:r]}, k_{[1:r]}, \omega,\hat{f}_{[1:r],\ta}^n, \hat{f}_{[1:r],\tb}^n)
\widehat{P}(y_{[1:2]}^n|x_{[1:2]}^n,f_{[1:r]}^n,b_{[1:r]}, k_{[1:r]}, \omega,\hat{f}_{[1:r],\ta}^n, \hat{f}_{[1:r],\tb}^n)\n
&=\widetilde{P}(x_{[1:2]}^n,f_{[1:r]}^n,b_{[1:r]}, k_{[1:r]}, \omega,\hat{f}_{[1:r],\ta}^n, \hat{f}_{[1:r],\tb}^n){p}(y_1^n|\hat{f}_{[1:r],\ta}^nx_1^n)p(y_2^n|\hat{f}_{[1:r],\tb}^nx_2^n)\n
&=\widetilde{P}(x_{[1:2]}^n,f_{[1:r]}^n,b_{[1:r]}, k_{[1:r]}, \omega,\hat{f}_{[1:r],\ta}^n, \hat{f}_{[1:r],\tb}^n){p}(y_1^n|{f}_{[1:r]}^nx_1^n)p(y_2^n|{f}_{[1:r]}^nx_2^n)\label{eq:apxb}\\
&=\widetilde{P}(x_{[1:2]}^n,f_{[1:r]}^n,b_{[1:r]}, k_{[1:r]}, \omega,y_{[1:2]}^n,\ \hat{f}_{[1:r],\ta}^n, \hat{f}_{[1:r],\tb}^n).
\end{align}
where the equation \eqref{eq:apxb} is due to the equality $\hat{f}_{[1:r],\ta}^n=\hat{f}_{[1:r],\tb}^n=f^n_{[1:r]}$ which is a result of \emph{indicator functions} in the definition of $\tilde{P}$ in \eqref{eq:apxa0}. This completes the proof of the approximation \eqref{aq:apxb2}.

\section{Proofs of mutual information bounds}\label{apx:d}
\subsection{Generalized version  of \cite[Lemma 2.7]{csiszar}}
\begin{lemma}[Modified version of {\cite[Lemma 2.7]{csiszar}}, c.f. \cite{zhang},{\cite[Problem 3.10]{csiszar}}]\label{le:csiszar}
For any two pmfs  $p_X$ and $p_{\hat{X}}$ on the same alphabet $\mx$, we have

\be
\left|H(X)-H(\hat{X})\right|\le \tv{p_X-p_{\hat{X}}}\log(|\mx|-1)+h_b\left(\tv{p_X-p_{\hat{X}}}\right). 
\ee
\end{lemma}
We now state a conditional extension of this lemma.
\begin{lemma}\label{le:gcsiszar}
For any pmf $p_Y$ on $\my$ and any two conditional pmfs $p_{X|Y}$ and $p_{\widehat{X}|Y}$ on the same alphabet $\mx$ 
%
\be
\left|H(X|Y)-H(\hat{X}|Y)\right|\le \tv{p_Yp_{X|Y}-p_Yp_{\widehat{X}|Y}}\log(|\mx|-1)+h_b\left(\tv{p_Yp_{X|Y}-p_Yp_{\widehat{X}|Y}}\right). 
\ee
\end{lemma}
\begin{proof}
\begin{align}
\left|H(X|Y)-H(\widehat{X}|Y)\right|&=\left|\sum_{y}p_Y(y)\left(H(X|Y=y)-H(\widehat{X}|Y=y)\right)\right|\n
                                             &\le\sum_{y}p_Y(y)\left|H(X|Y=y)-H(\widehat{X}|Y=y)\right|\n
                                             &\le \sum_{y}p_Y(y)\left[\tv{p_{X|Y=y}-p_{\hat{X}|Y=y}}\log(|\mx|-1)+h\left(\tv{p_{X|Y=y}-p_{\hat{X}|Y=y}}\right)\right]\label{eq:apxd0}\\
                                             &= \tv{p_Yp_{X|Y}-p_Yp_{\widehat{X}|Y}}\log(|\mx|-1)+\sum_{y}p_Y(y)h\left(\tv{p_{X|Y=y}-p_{\hat{X}|Y=y}}\right)\n
                                             &\le \tv{p_Yp_{X|Y}-p_Yp_{\widehat{X}|Y}}\log(|\mx|-1)+h_b\left(\sum_{y}p_Y(y)\tv{p_{X|Y=y}-p_{\hat{X}|Y=y}}\right)\label{eq:apxd1}\\
                                             &= \tv{p_Yp_{X|Y}-p_Yp_{\widehat{X}|Y}}\log(|\mx|-1)+h_b\left(\tv{p_Yp_{X|Y}-p_Yp_{\widehat{X}|Y}}\right),  \label{eq:apxd3}                                           
\end{align}
where \eqref{eq:apxd0} follows from Lemma \ref{le:csiszar} and \eqref{eq:apxd1} follows from Jensen inequality for the concave function $h_b$(.).\end{proof}
\subsection{Proof of Lemma \ref{le:c1}}
The proof is similar to the one given in \cite[Lemma VI. 3]{cuff-trans} for the unconditional case of the Lemma \ref{le:c1}. First we use the first part of Lemma \ref{le:total} to obtain the closeness of $p(w_q,z)$ and $\widehat{p}_q(w_q|z)p(z)$ in total variation. In other words we have
\[ \tv{p(w_q,z)-\widehat{p}_q(w_q|z)p(z)}\le \epsilon.\]
Let $\hat{W}^n$ be a random variable such that $p_{\hat{W}^n,Z}(w^n,z)=p(z)\prod_{q=1}^n p_q(w_q|z)$. Then Lemma \ref{le:gcsiszar} implies that
\begin{align}
\left|H(W_q|Z)-H(\hat{W}_q|Z)\right|&\le \epsilon\log|\mw|+h_b(\epsilon),\n
\left|H(W^n|Z)-H(\hat{W}^n|Z)\right|&\le \epsilon\log|\mw|^n+h_b(\epsilon)=  n\epsilon \log|\mw|+h_b(\epsilon).\label{eq:apxd-1}
\end{align}
Now we have
\begin{align}
\sum_{q=1}^n I(W_q;W^{q-1}|Z)&=\sum_{q=1}^n H(W_q|Z)- H(W^n|Z)\n
                                                  &=\sum_{q=1}^n \left[H(W_q|Z)-H(\hat{W}_q|Z)\right]+H(\hat{W}^n|Z)-H(W^n|Z)\label{eq:apxd4}\\
                                                  &\le 2n\epsilon\log|\mw|+(n+1)h_b(\epsilon) \label{eq:apxd5}
\end{align}
where \eqref{eq:apxd4} follows from the fact that $H(\hat{W}^n|Z)=\sum_{q=1}^n H(\hat{W}_q|Z)$, because of $p_{\hat{W}^n,Z}(w^n,z)=p(z)\prod_{q=1}^n p_q(w_q|z)$ and \eqref{eq:apxd5} is a result of \eqref{eq:apxd-1}. This concludes the proof.
\subsection{Proof of Lemma \ref{le:c2}}
The proof is similar to the proof of Lemma \ref{le:c1}. First, using the first part of Lemma \ref{le:total} we have
\[\tv{p(y_q,x_q)-\widehat{p}(y|x)p(x)}=\tv{p(y_q,x_q)-\widehat{p}(y_q|x_q)p(x_q)}\le \epsilon.\]
Let $\hat{Y}^n$ be a random variable such that $p_{X^n,\hat{Y}^n}(x^n,y^n)=\prod_{q=1}^n p(x_q)\widehat{p}(y_q|x_q)$. 
 Observe that $H(\hat{Y}_q|X^n)=H(\hat{Y}_q|X_q)=H(\hY|X)$, where $(X,\hY)$ is distributed according to $p(x)\hat{p}(y|x)$. Then Lemma \ref{le:gcsiszar} implies that 
\be
\begin{split}
\left|H(Y_q|X_q)-H(\hat{Y}_q|X_q)\right|&\le \epsilon\log|\my|+h_b(\epsilon),\\
\left|H(\hat{Y}_q|X^n)-H(Y_q|X^n)\right|&\le \epsilon\log|\my|+h_b(\epsilon).
\end{split}\label{eq:apxd20}
\ee
We have
\begin{align}
I(X_{[\sim q]};Y_q|X_q)&=H(Y_q|X_q)-H(Y_q|X^n)\n
                                   &\le \left|H(Y_q|X_q)-H(\hat{Y}_q|X_q)\right| +\left|H(\hat{Y}_q|X_q)-H(Y_q|X^n)\right|\n
                                   &=  \left|H(Y_q|X_q)-H(\hat{Y}_q|X_q)\right| +\left|H(\hat{Y}_q|X^n)-H(Y_q|X^n)\right|\n
                                   &\le 2\left(\epsilon\log|\my|+h_b(\epsilon)\right)
                                   ,\label{eq:apxd21}
\end{align}
where \eqref{eq:apxd21} follows from \eqref{eq:apxd20}. This completes the proof of \eqref{eq:le-c2-1}.

Next we prove \eqref{eq:le-c2-2}. First we note that $(X_Q,\hY_Q)$ is distributed according to $p(x)\hat{p}(y|x)$, because $(X^n,\hY^n)$ is jointly i.i.d. according to  $p(x)\hat{p}(y|x)$. Also, by \cite[Lemma VI.2]{cuff:synthesis} we have the closeness between $p(x_Q,y_Q)$ and $p(x)\hat{p}(y|x)$, that is, 
\[
\tv{p(x_Q,y_Q)-p(x)\hat{p}(y|x)}\le\epsilon.
\]
Then Lemma \ref{le:gcsiszar} implies that
\be
\left|H(Y_Q|X_Q)-H(\hat{Y}|X)\right|\le \epsilon\log|\my|+h_b(\epsilon).\label{eq:apxd22}
\ee
Next consider
\be
\begin{split}
\left|H(Y_Q|X_Q,Q)-H(\hY|X)\right|&=\sum_q p_Q(q)\left|H(Y_q|X_q)-H(\hY|X)\right|\\
                                                     &\le\sum_q p_Q(q) \left(\epsilon\log|\my|+h_b(\epsilon)\right)\\
                                                     &= \epsilon\log|\my|+h_b(\epsilon),
\end{split}
\label{eq:apxd23}
\ee
where we used \eqref{eq:apxd20}. Finally, combining \eqref{eq:apxd22} and \eqref{eq:apxd23} implies \eqref{eq:le-c2-2}.

\section{Cardinality Bounds}\label{apx:cardinality}
The cardinality bounds can be proved inductively using the support lemma \cite[Appendix C]{elgamal}. Here we provide the sketch of the proof.  Assume that we have reduced the cardinalities of $F_1, F_2, \cdots, F_{i-1}$. We prove a cardinality bound on $F_i$. For simplicity we only write the case of $i>1$; the case of $i=1$ is similar. Take some arbitrary $q(f_{[1:r]},x_{[1:2]},y_{[1:2]})$ with $F_j, j\in[1:r]$ taking values from finite (but arbitrarily large) sets. In the statement of the support lemma, we consider $\mathscr{P}$ to be the set of all pmfs 
$\pi(f_{[1:i-1]},f_{[i+1:r]},x_{[1:2]},y_{[1:2]})$
on 
$\mathcal{F}_{1}\times\mathcal{F}_{2}\times\cdots\times\mathcal{F}_{i-1}\times
\mathcal{F}_{i+1}\times\cdots\times\mathcal{F}_{r}\times\mathcal{X}_{1}\times\mathcal{X}_{2}\times\mathcal{Y}_{1}\times\mathcal{Y}_{2}$
that 
satisfy the following 
\begin{itemize}
\item $\pi(x_2|f_{[1:i-1]}x_1)=q(x_2|f_{[1:i-1]}x_1)$ if $i$ is odd; or $\pi(x_1|f_{[1:i-1]}x_2)=q(x_1|f_{[1:i-1]}x_2)$ if $i$ is even;
\item
For any $j>i$: 
 $F_j-F_{[1:i-1]}F_{[i+1:j-1]}X_1-X_2, \ \mbox{if $j$ is odd;}$ $F_j-F_{[1:i-1]}F_{[i+1:j-1]}X_2-X_1, \ \mbox{if $j$ is even; }
Y_1-F_{[1:i-1]}F_{[i+1:r]} X_1-X_2Y_2,$ and $Y_2-F_{[1:i-1]}F_{[i+1:r]} X_2-X_1Y_1$.
\end{itemize}
This set is compact and connected. To see its connectedness, for simplicity consider the special case of $r=3, i=2$; the proof for general case is similar. 
$\mathscr{P}$ is the set of $\pi(x_{[1:2]}y_{[1:2]}f_1f_3)$ that factorize as follows:
$$\pi(x_{[1:2]}y_{[1:2]}f_1f_3)=\pi(x_2f_{1})q(x_1|f_{1}x_2)\pi(f_3|x_1f_1)\pi(y_1|x_1f_{1}f_{3})\pi(y_2|x_2f_{1}f_{3}).$$
Given $\pi_1(x_{[1:2]}y_{[1:2]}f_1f_3)$ and $\pi_2(x_{[1:2]}y_{[1:2]}f_1f_3)$ of the above form, we can continously move from $\pi_1$ to $\pi_2$ in several steps, by first moving from $\pi_1(x_{[1:2]}y_{[1:2]}f_1f_3)$ to
$$\pi_2(x_2f_{1})q(x_1|f_{1}x_2)\pi_1(f_3|x_1f_1)\pi_1(y_1|x_1f_{1}f_{3})\pi_1(y_2|x_2f_{1}f_{3}).$$
where the first term $\pi_1(x_2f_{1})$ is replaced with $\pi_2(x_2f_{1})$ via continous moves. We can then replace the term $\pi_1(f_3|x_1f_1)$ with $\pi_2(f_3|x_1f_1)$, etc.

Then we consider the following continuous functions on $\mathscr{P}$. Given any $(x_{[1:2]},y_{[1:2]}, f_1,\cdots, f_{i-1})$ and pmf $\pi$ on $\mathscr{P}$ we define
$$g_{x_{[1:2]},y_{[1:2]},f_{[1:i-1]}}(\pi)=\mathbb{P}_{\pi}[X_{[1:2]}=x_{[1:2]}, Y_{[1:2]}=y_{[1:2]},  F_{[1:i-1]}=f_{[1:i-1]}].$$
Further we define three more functions:
\begin{align}
g_{1}(\pi)&=H(X_1|F_{1:i-1}F_{i+1:r}X_2), \\
g_{2}(\pi)&=H(X_2|F_{1:i-1}F_{i+1:r}X_1), \\
g_{3}(\pi)&=H(Y_{[1:2]}|F_{1:i-1}F_{i+1:r}X_{[1:2]}).
\end{align}
We consider $g_{x_{[1:2]},y_{[1:2]},f_{[1:i-1]}}$ for all values of $x_{[1:2]},y_{[1:2]},f_{[1:i-1]}$ except for one arbitrary tuple $(x^*_{[1:2]},y^*_{[1:2]},f^*_{[1:i-1]})$, giving us $|\mathcal{X}_{1}||\mathcal{X}_{2}||\mathcal{Y}_{1}||\mathcal{Y}_{2}|\prod_{j=1}^{i-1}|\mathcal{F}_{j}|-1$ functions. Thus in total we have $|\mathcal{X}_{1}||\mathcal{X}_{2}||\mathcal{Y}_{1}||\mathcal{Y}_{2}|\prod_{j=1}^{i-1}|\mathcal{F}_{j}|+2$ functions. Applying the support lemma, we can reduce the cardinality of $F_i$ to $|\mathcal{X}_{1}||\mathcal{X}_{2}||\mathcal{Y}_{1}||\mathcal{Y}_{2}|\prod_{j=1}^{i-1}|\mathcal{F}_{j}|+2$ by finding some $p(x_{[1:2]}, y_{[1:2]}, f_{[1:r]})$ such that
$$p(x_{[1:2]}, y_{[1:2]}, f_{[1:i-1]})=q(x_{[1:2]}, y_{[1:2]}, f_{[1:i-1]}),$$ 
$$H_{p}(X_1|F_{1:i-1}F_{i+1:r}X_2F_i)=H_{q}(X_1|F_{1:i-1}F_{i+1:r}X_2F_i)$$
$$H_p(X_2|F_{1:i-1}F_{i+1:r}X_1F_i)=H_q(X_2|F_{1:i-1}F_{i+1:r}X_1F_i)$$
and 
$$H_p(Y_{[1:2]}|F_{1:i-1}F_{i+1:r}X_{[1:2]}F_i)=H_q(Y_{[1:2]}|F_{1:i-1}F_{i+1:r}X_{[1:2]}F_i).$$
Further the resulting $p(x_{[1:2]}, y_{[1:2]}, f_{[1:i-1]},  f_{[i+1:r]}|f_i)$ is in $\mathscr{P}$, implying the Markov chain equations for $j\geq i$: $F_j-F_{[1:j-1]}X_1-X_2, \ \mbox{if $j$ is odd;} F_j-F_{[1:j-1]}X_2-X_1, \ \mbox{if $j$ is even;} 
Y_1-F_{[1:r]} X_1-X_2Y_2,$ and $Y_2-F_{[1:r]} X_2-X_1Y_1$. The first condition imposed on $\mathscr{P}$ implies the Markov chains for $j=i$, whereas the second condition implies it for $j>i$. Since we are preserving 
$p(x_{[1:2]}, y_{[1:2]}, f_{[1:i-1]})$, the chains also hold for $j<i$. Further we get that 
$H(X_1|X_2), H(X_2|X_1)$, $H(Y_{[1:2]}|X_{[1:2]})$ and $I(F_1;Y_{[1:2]}|X_{[1:2]})$ are preserved. 

\section{Rate elimination}\label{apx:elimination}
We eliminate the rates $(\tR_1,\cdots,\tR_r)$ and $(R_1,\cdots,R_r)$ in few steps.

\underline{Step 1: Relaxing the implicit positivity constraints on $\tR_i,i\in[1:r]$. } 
First we want to eliminate the rates $(\tR_1,\cdots,\tR_r)$ from \eqref{eq:c1}-\eqref{eq:c3V2} and \eqref{eq:c44}. However we also have the implicit constraints  $\tR_i\ge0, i\in[1:r]$. Nevertheless, we show that these constraints are redundant. To do this, we show that if $(R_0,R_1,\cdots,R_r,\tR_1,\cdots,\tR_r)$ satisfies \eqref{eq:c1}-\eqref{eq:c3V2} and \eqref{eq:c44} for some r.v.s $F_{[1:r]}$ and the rates$(\tR_1,\cdots,\tR_r)$ (which are not necessarily positive), then there exists a r.v.s $\bar{F}_{[1:r]}$ and $\bar{R}_i\ge 0, , i\in[1:r]$ such that $(R_0,R_1,\cdots,R_r,\bar{R}_1,\cdots,\bar{R}_r)$ satisfies \eqref{eq:c1}-\eqref{eq:c3V2} and \eqref{eq:c44} for $\bar{F}_{[1:r]}$ instead of $F_{[1:r]}$ and $\bar{R}_i$ instead of $\tR_i$.
 Let $W_i, i\in[1:r]$ be  r.v.s with entropies $H(W_i)>|\tR_i|$. Further assume that $W_i, i\in[1:r]$ are  independent of each other and also independent of all other r.v.'s, i.e. $(F_{[1:r]},X_{[1:2]},Y_{[1:2]})$. Let $\bar{R_i}=\tR_i+H(W_i)$ and $\bar{F_i}=(F_i,W_i)$. 
It is clear that $\bar{R}_i>0, \forall i$. Now it can easily shown that$(R_0,R_1,\cdots,R_r,\bar{R}_1,\cdots,\bar{R}_r)$ satisfies \eqref{eq:c1}-\eqref{eq:c3V2} and \eqref{eq:c44} for $\bar{F}_{[1:r]}$, using the independence of $W_{[1:r]}$ from all other r.v.'s and the  fact that $(R_0,R_1,\cdots,R_r,\tR_1,\cdots,\tR_r)$ satisfies  \eqref{eq:c1}-\eqref{eq:c3V2} and \eqref{eq:c44} for $F_{[1:r]}$.

\underline{Step 2: Eliminating the rates $\tR_i,i\in[1:r]$. } Without loss of generality, we can assume that the constraints \eqref{eq:c1} and \eqref{eq:c2} hold with equality, because we can decrease the rates $\tR_i, i\in[1:r]$ to get equality in the constraints \eqref{eq:c1} and \eqref{eq:c2} without disturbing the other constraints. In this case, we have
\begin{align*}
\tR_1&=H(F_1|X_2)-R_0-R_1,\\
\tR_i&=H(F_i|X_{(i+1)_2}F_{[1:i-1]}),~~ \mbox{for $i\in[2:r]$}.
\end{align*}  
Substituting these equalities in \eqref{eq:c3V1}, \eqref{eq:c3V2} and \eqref{eq:c44} gives the following constraints for $i\in[1:r]$,
\begin{align}
R_i&\ge I(X_{(i)_2};F_i|F_{[1:i-1]}X_{(i+1)_2}),\label{eq:constraint-1}\\
R_0+\sum_{t=1}^i R_t&\ge \sum_{t=1}^i I(F_t;X_{(t)_2}Y_{[1:2]}|X_{(t+1)_2}F_{[1:t-1]})\n
                                   &=\sum_{t=1}^i I(F_t;X_{(t)_2}|X_{(t+1)_2}F_{[1:t-1]})+\sum_{t=1}^i I(F_t;Y_{[1:2]}|X_{[1:2]}F_{[1:t-1]})\n
                                   &=I(F_{[1:i]};Y_{[1:2]}|X_{[1:2]})+\sum_{t=1}^i I(F_t;X_{(t)_2}|X_{(t+1)_2}F_{[1:t-1]}).\label{eq:constraint-2}
\end{align}
\underline{Step 3: Eliminating the rates $R_i,i\in[1:r]$. } In this step we want to eliminate the rates $(R_1,\cdots,R_r)$ from \eqref{eq:rate-split}, \eqref{eq:constraint-1} and \eqref{eq:constraint-2}. This can be done using Fourier-Motzkin elimination (FME). Applying FME gives the following constraints on $(R_0,R_{12},R_{21})$:
\begin{align}
R_{12}&\ge I(X_1;\F|X_2),\n
R_{21}&\ge I(X_2;\F|X_1),\n
R_0+R_{12}&\ge I(X_1;\F|X_2)+I(F_1;Y_{[1:2]}|X_{[1:2]}),\n
R_0+R_{12}+R_{21}&\ge I(X_1;\F|X_2)+I(X_2;\F|X_1) +I(F_{[1:i]};Y_{[1:2]}|X_{[1:2]}),~~~\mbox{for $i\in[2:r]$}.
\end{align}
Finally we note that the last constraints for $i\in[2:r-1]$ are redundant due to the constraint corresponding to $i=r$. 

\section{Proof of Markov chains in \eqref{eqn:apx-F}}\label{apx:c}
We know that for any code the following Markov chain conditions hold
\be\label{eq:apxc0}\begin{split}
C_i-&\omega C_{[1:i-1]}X_1^n-X_2^n, \ \mbox{if $i$ is odd,} \\
C_i-&\omega C_{[1:i-1]}X_2^n-X_1^n, \ \mbox{if $i$ is even,} \\
&Y_1^n-\omega C_{[1:r]} X_1^n-X_2^nY_2^n,\\
&Y_2^n-\omega C_{[1:r]} X_2^n-X_1^nY_1^n.
\end{split}\ee
The Markov chain $C_i-\omega C_{[1:i-1]}X_1^{q:n}X_2^{1:q-1}-X_{2,q}$ for odd $i$ holds because
\begin{align}
I(C_i;X_{2,q}|\omega C_{[1:i-1]}X_1^{q:n}X_2^{1:q-1})&\le I(C_i X_1^{1:q-1};X_{2,q}|\omega C_{[1:i-1]}X_1^{q:n}X_2^{1:q-1})\n
                                                                                     &=   I(X_1^{1:q-1};X_{2,q}|\omega C_{[1:i-1]}X_1^{q:n}X_2^{1:q-1})\label{eq:apxc1}\\
                                                                                     &\le I(X_1^{1:q-1};X_{2}^{q:n}|\omega C_{[1:i-1]}X_1^{q:n}X_2^{1:q-1})\n
                                                                                     &=   0, \label{eq:apxc2}
\end{align}
where \eqref{eq:apxc1} follows from the first Markov chain of \eqref{eq:apxc0}, and the \eqref{eq:apxc2} follows from Lemma \ref{le:markov} provided at the end of this appendix. 
Similarly the Markov chain $C_i-\omega C_{[1:i-1]}X_1^{q+1:n}X_2^{1:q}-X_{1,q}$ for even $i$ holds. Next, we show that the Markov chain
$Y_{1,q}-\omega C_{[1:r]}X_1^{q:n}X_2^{1:q-1}-X_{2,q}Y_{2,q}$ holds.
\begin{align}
I(Y_{1,q};X_{2,q}Y_{2,q}|\omega C_{[1:r]}X_1^{q:n}X_2^{1:q-1})&\le I(Y_{1,q}X_1^{1:q-1};X_{2,q}Y_{2,q}|\omega C_{[1:r]}X_1^{q:n}X_2^{1:q-1})\n
                                                                                                    &=   I(X_1^{1:q-1};X_{2,q}Y_{2,q}|\omega C_{[1:r]}X_1^{q:n}X_2^{1:q-1})\label{eq:apxc3}\\
                                                                                                    &\le I(X_1^{1:q-1};X_{2}^{q:n}Y_{2,q}|\omega C_{[1:r]}X_1^{q:n}X_2^{1:q-1})\n
                                                                                                    &=   I(X_1^{1:q-1};X_{2}^{q:n}|\omega C_{[1:r]}X_1^{q:n}X_2^{1:q-1})\label{eq:apxc4}\\
                                                                                                    &=   0,\label{eq:apxc5}
\end{align}
where \eqref{eq:apxc3} follows from the third Markov chain of \eqref{eq:apxc0}, \eqref{eq:apxc4} follows from the last Markov chain of \eqref{eq:apxc0} and the \eqref{eq:apxc5} follows from Lemma \ref{le:markov}.
Similarly the Markov chain $Y_{2,q}-\omega C_{[1:r]}X_1^{q+1:n}X_2^{1:q}-X_{1,q}Y_{1,q}$ holds.

\begin{lemma} \label{le:markov}For any set of random variables satisfying the Markov chain constraints of \eqref{eq:apxc0}, the following holds:
\be
                                              \forall q, i:\qquad            I(X_1^{1:q-1};X_{2}^{q:n}|\omega C_{[1:i]}X_1^{q:n}X_2^{1:q-1})=0. \label{eq:apxc30}
\ee

\end{lemma}
\begin{proof}
We prove the lemma by induction on $i$. For $i=0$, we have $I(X_1^{1:q-1};X_{2}^{q:n}|\omega X_1^{q:n}X_2^{1:q-1})=0$ because $X_{[1:2]}^n$ is i.i.d.\ and is independent of the common randomness $\omega$. Suppose that the statement of the lemma holds for $i=j-1$. For $i=j$ we proceed as follows: 
\begin{itemize}
\item If $j$ is odd, we have
\begin{align}
I(X_1^{1:q-1};X_{2}^{q:n}|\omega C_{[1:j]}X_1^{q:n}X_2^{1:q-1})&\le I(C_jX_1^{1:q-1};X_{2}^{q:n}|\omega C_{[1:j-1]}X_1^{q:n}X_2^{1:q-1})\n
                                                                                                      &=   I(X_1^{1:q-1};X_{2}^{q:n}|\omega C_{[1:j-1]}X_1^{q:n}X_2^{1:q-1})=0, 
\end{align}
where in the last step we use the first Markov chain of \eqref{eq:apxc0} and the induction assumption.
\item If $j$ is even, we have
\begin{align}
I(X_1^{1:q-1};X_{2}^{q:n}|\omega C_{[1:j]}X_1^{q:n}X_2^{1:q-1})&\le I(X_1^{1:q-1};C_jX_{2}^{q:n}|\omega C_{[1:j-1]}X_1^{q:n}X_2^{1:q-1})\n
                                                                                                      &=   I(X_1^{1:q-1};X_{2}^{q:n}|\omega C_{[1:j-1]}X_1^{q:n}X_2^{1:q-1})=0,
\end{align}
where in the last step we use the second Markov chain of \eqref{eq:apxc0} and the induction assumption.
\end{itemize}
This completes the induction proof.
\end{proof}

\section{Converse Proof of Theorem \ref{thm:emp}}\label{apx:conv}
Assume $(R_{12},R_{21})$ is a pair of achievable rate. Consider a sequence of coordination codes that achieves $(R_{12},R_{21})$. Take a random variable $Q$ uniform on $[1:n]$ and independent of all other random variables. Define $F_i=C_i X_{1}^{Q+1:n}X_2^{1:Q-1}Q$ for $1\le i\le r$ and $X_i=X_{iQ},Y_i=Y_{iQ}$ for $i=1,2$\footnote{Following the standard definition of empirical coordination code we assume that there is not any common randomness, that is $\omega$ is a constant random variable. See \cite{cuff:thesis} and Remark \ref{re:cor}.}. In the first step of the proof, we show the Markov chain conditions given in the definition of $T(r)$ are satisfied by this choice of auxiliary r.v.'s. The proof of this fact is similar to the one given in the Appendix \ref{apx:c} and hence it is omitted here.
\begin{align}
nR_{12}&\ge \sum_{i:odd}H(C_i)\n&
\ge \sum_{i:odd}I(C_i;X_1^n|C_{[1:i-1]}X_2^n)\n
&= \sum_{i=1}^rI(C_i;X_1^n|C_{[1:i-1]}X_2^n)\label{eq:aA1b}\\
&=I(C_{[1:r]};X_1^n|X_2^n)\n
&=\sum_{q=1}^n I(C_{[1:r]};X_{1,q}|X_{1}^{q+1:n}X_2^n)\n
&=\sum_{q=1}^n I( C_{[1:r]}X_{1}^{q+1:n}X_{2,\sim q};X_{1,q}|X_{2,q})\label{eq:aA0b}\\
&\ge \sum_{q=1}^n I( C_{[1:r]}X_{1}^{q+1:n}X_{2}^{1:q-1};X_{1,q}|X_{2,q})\n
&=nI( C_{[1:r]}X_{1}^{Q+1:n}X_{2}^{1:Q-1};X_{1,Q}|X_{2,Q},Q) \n
&=nI( C_{[1:r]}X_{1}^{Q+1:n}X_{2}^{1:Q-1}Q;X_{1,Q}|X_{2,Q})\label{eq:a3b}\\
&=nI(\F;X_{1}|X_{2})\label{eq:a4b},
\end{align}
where \eqref{eq:aA1b} follows from the Markov chain $C_i-C_{[1:i-1]}X_2^nX_1^n$ for even $i$, \eqref{eq:aA0b} follows from the fact that $X_{1q},X_{2q}$ are i.i.d. repetitions and \eqref{eq:a3b} follows from the fact that $Q$ is independent of $(X_{1,Q},X_{2,Q})$ (See \cite{cuff:thesis}). The inequality $R_{21}\ge I(\F;X_{2}|X_{1})$ can be proved similarly.

The definition of coordination code implies that
\be
\e\tv{\tilde{\mathbf{p}}_{X_{[1:2]}^nY_{[1:2]}^n}-q_{X_{[1:2]}Y_{[1:2]}}}\rightarrow 0.
\ee
This yields that
\be
\tilde{\mathbf{p}}_{X_{[1:2]}^nY_{[1:2]}^n}\rightarrow q_{X_{[1:2]}Y_{[1:2]}}.
\ee
In the other side, it is shown in \cite{cuff:thesis} that $\e\tilde{\mathbf{p}}_{X_{[1:2]}^nY_{[1:2]}^n}=\tilde{p}_{X_{[1:2],Q},Y_{[1:2],Q}}$ where $\tilde{p}$ is the induced pmf by the code. Therefore $\tilde{p}_{X_{[1:2],Q},Y_{[1:2],Q}}$ tends to $q_{X_{[1:2]}Y_{[1:2]}}$. Now the closedness of  the coordination rate region completes the proof.  

\end{document}